\documentclass{article}

\usepackage[a4paper,top=2.54cm,bottom=2.54cm,left=3.17cm,right=3.17cm,%
            includehead,includefoot]{geometry}

\usepackage{amsmath,amssymb,amsfonts,amsthm}
\usepackage{graphicx}
\usepackage{subfigure}
\usepackage{float}
\usepackage{xcolor}
\usepackage[numbers,square,sort&compress]{natbib}
\usepackage[hidelinks]{hyperref}
\usepackage{booktabs}
\usepackage{colortbl}
\usepackage{caption}
\usepackage{enumitem}
\usepackage{epstopdf}
\usepackage{algorithm}
\usepackage{algpseudocode}
\usepackage{array}
\usepackage{bbding}
\usepackage{fancyhdr} % 页眉
\usepackage{fancyvrb} % 摘录 Verbatim 环境
\usepackage{longtable}
\usepackage{listings}
\usepackage{rotating,rotfloat} % 提供 sidewayfigure 和 sidewaystable 环境横排图表
\usepackage{yhmath} % \wideparen 弧 \adots
\usepackage{authblk}
\usepackage{multirow}
\usepackage{braket}
\usepackage{mathrsfs}
\usepackage{amsmath}
\usepackage{mathrsfs}
\usepackage{mathtools}
\graphicspath{{./}{figures/}}  % 插图所在文件夹
\usepackage{fancyhdr}
\newcommand{\mytitle}{Quantification of the Resource Theory of Imaginarity Revisited}

\renewcommand{\sectionmark}[1]{}
\renewcommand{\subsectionmark}[1]{}

\allowdisplaybreaks

\newtheorem{theorem}{Theorem}[section]
\newtheorem{lemma}{Lemma}[section]
\newtheorem{corollary}[theorem]{Corollary}
\newtheorem{proposition}[theorem]{Proposition}
\newtheorem{definition}{Definition}[section]

\newtheorem{remark}{Remark}

\newcommand{\bbm}{\begin{bmatrix}}
\newcommand{\ebm}{\end{bmatrix}}

\begin{document}
\title{\Large Quantification of the Resource Theory of Imaginarity Revisited}

\author{Yue Han}
\author{Naihong Hu%
	\thanks{N.H. acknowledges funding from the Fundamental and
		Interdisciplinary Disciplines Breakthrough Plan of the Ministry
		of Education of China (JYB2025XDXM112), and the Science and
		Technology Commission of Shanghai Municipality
		(Grant No.~22DZ2229014).}}

\affil{\small School of Mathematical Sciences, MOE-KLMEA \&
	SH-KLPMMP, East China Normal University,\\
	Minhang Campus, 500 Dongchuan Road, Shanghai 200241,
	P.~R.~China\\
	\textit{E-mail addresses:}
	\texttt{52280155013@stu.ecnu.edu.cn},
	\texttt{nhhu@math.ecnu.edu.cn}\\
	\textit{ORCID:}
	Y.~Han,
	\href{https://orcid.org/0009-0006-9015-0978}
	{\textcolor{black}{0009-0006-9015-0978}};
	N.~Hu,
	\href{https://orcid.org/0000-0001-8162-8155}
	{\textcolor{black}{0000-0001-8162-8155}}}
\date{}
\maketitle
\begin{abstract}
This work studies the decay of three imaginarity measures, namely the $l_1$-norm‑based imaginarity, the robustness of imaginarity, and the relative‑entropy‑based imaginarity, under typical noise channels. For arbitrary single‑qubit pure states subject to real quantum channels, we prove an exact attenuation theorem: the $l_1$-norm‑based imaginarity and the robustness exhibit identical attenuation governed by a common channel‑dependent factor, whereas the relative‑entropy‑based imaginarity additionally depends on the input‑state orientation relative to the reference basis. We extend the investigation to two‑qubit entangled and dual‑rail states, establishing a quantitative link between photon‑loss probability and residual imaginarity for photonic dual‑rail encodings. For separable two‑qubit states, we define a real‑operation resource preorder and show that $\Omega_{++}=|+\rangle\langle+|\otimes|+\rangle\langle+|$ is operationally maximal. We further distinguish operational de‑imaginary power from reference‑state imaginarity loss. For several real two‑qubit channels we derive exact operational de‑imaginary powers, while for other channels we provide analytic reference‑state losses serving as benchmarks and rigorous lower bounds.
\end{abstract}



\section{Introduction}
Quantum information technology has experienced rapid development in recent years. As a cutting-edge technology grounded in quantum mechanics, it integrates multiple disciplines and has spurred increasing research attention toward the control and manipulation of microscopic quantum systems. Quantum Resource Theory (QRT) provides a rigorous mathematical framework for quantifying and manipulating quantum mechanical resources — including entanglement \cite{Horodecki09}, superposition  \cite{Theurer20117}, steering \cite{Gallego2015}, magic  \cite{Veitch2014}\cite{Howard2017}, and nonlocality \cite{JI2014} — and thus identifies the information-processing tasks achievable with these resources \cite{Brandão15}\cite{Chitambar19}.

In the traditional framework of quantum mechanics, complex numbers play an important role in any canonical form \cite{WK2012}\cite{Hardy2012}\cite{Aleksandrova13}. Their imaginary parts are not merely a formal mathematical representation but rather a key element with profound physical significance. Based on the imaginary terms present in the density matrix of a quantum state, Hickey and Gour \cite{Hickey18} proposed the resource theory of imaginarity and introduced the concept of an imaginarity measure. Focusing on two specific imaginarity measures — the geometric imaginarity and the robustness of imaginarity — Wu et al. \cite{Wu21}\cite{WKR21} quantified quantum imaginarity and provided an operational interpretation for the quantification of imaginarity in state conversion problems. Quantifying and analyzing imaginarity is crucial for advancing quantum information technology, as the same quantum state often yields different quantitative results when evaluated using different methods under varying physical contexts. Xue et al. \cite{Xue21} further explored the quantification of quantum state imaginarity by introducing two new measures: the weight of imaginarity and the relative entropy of imaginarity, while also investigating the decay of imaginarity when a single-qubit quantum state passes through a quantum channel. More recently, Chen et al. \cite{Chen23} proposed imaginarity measures based on the $l_{1}$ norm and convex functions, and studied the changes in the order of single-qubit quantum states after transmission through a quantum channel.

Since the imaginary part of a quantum state’s density matrix only appears in its off-diagonal elements, the quantum resource theory of imaginarity is closely associated with the quantum resource theory of coherence \cite{Winter16}\cite{Chitambar16}\cite{Streltsov17}\cite{Mani15}. For instance, measures based on relative entropy were first proposed and studied within the framework of coherence resource theory by Baumgratz et al. \cite{Baumgratz14}. As a transmission channel for quantum information, quantum channels play a pivotal role in fields such as quantum computing and quantum communication. They can be employed to transmit quantum keys \cite{Ekert91} and realize quantum teleportation \cite{Bennett93}, among other applications. Notably, the imaginarity of a quantum state is often altered when it passes through a quantum channel; thus, studying the impact of quantum channels on quantum imaginarity is of great theoretical and practical significance. By exploring the role of the imaginarity resource in quantum channels, we can facilitate the development of more secure quantum communication systems.

In this paper, we first study how different imaginarity measures decay for arbitrary single‑qubit pure states under dephasing, generalized amplitude damping, and phase‑amplitude damping channels. We consider the $l_1$-norm‑based imaginarity measure, the robustness of imaginarity, and the relative entropy of imaginarity. We establish an exact attenuation theorem for arbitrary real single‑qubit quantum channels, showing that the $l_1$-norm‑based imaginarity and the robustness of imaginarity are equal and undergo exactly the same channel‑dependent attenuation. For the relative entropy of imaginarity, we show that the initial resource is determined by the parameter $A=|\langle\psi^*|\psi\rangle|$, whereas its decay under a fixed basis‑dependent channel generally requires an additional parameter describing the orientation of the input state relative to the reference basis.

We then extend the state‑decay analysis to representative two‑qubit entangled states and dual‑rail states under two‑qubit noise channels. For a photonic dual‑rail qubit subject to symmetric amplitude damping, we interpret the damping process as photon loss and show that the vacuum component lies outside the logical code space and can be identified as an erasure event. This provides a concrete physical application of the analytical decay formulas and a direct quantitative criterion for the preservation of imaginarity under photon loss.

For separable two‑qubit systems, we introduce a real‑operation preorder and an operational notion of a maximal separable state. We prove that the product state $\Omega_{++}=|+\rangle\langle+|\otimes|+\rangle\langle+|$ is operationally maximal and use it as a common, locally preparable reference state. On this basis, we distinguish three related channel quantities: the imaginary power on separable states, the operational de‑imaginary power, and the imaginarity loss evaluated on the fixed reference state $\Omega_{++}$. For the real channels considered in this work, the imaginary power is zero. We obtain exact operational de‑imaginary powers, based on the robustness and relative entropy of imaginarity, for the two‑qubit phase‑damping, phase‑flip, and bit‑flip channels. For the amplitude‑damping, phase‑amplitude‑damping, bit‑phase‑flip, and depolarizing channels, we calculate analytically tractable reference‑state losses, which provide common benchmarks and rigorous lower bounds on the corresponding operational de‑imaginary powers.

\section{Preliminaries}

\subsection{Framework of the Resource Theory of Imaginarity}
We review the fundamentals of the resource theory of imaginarity with reference to Refs. \cite{Hickey18}\cite{Wu21}\cite{WKR21}. From the basic assumptions of quantum mechanics: any isolated physical system is associated with a complex inner product vector space (i.e., a Hilbert space) called the state space of the system, and the system is completely described by a state vector, which is a unit vector in the system's state space. Now, we consider a quantum system associated with a $d$-dimensional complex Hilbert space $\mathscr{H}$, and $\{\ket{j}\}_{j=0}^{d-1}$ is an orthonormal basis on $\mathscr{H}$. Let $\rho$ be a density operator on the Hilbert space $\mathscr{H}$. The set of all density operators on the Hilbert space $\mathscr{H}$ is denoted by $\mathscr{D}(\mathscr{H})$.

In the standard framework of the theory of quantum resources, free states and free operations are the two core fundamental concepts. Free states are defined as quantum states carrying no target resource, while free operations refer to quantum operations that cannot generate the target resource from scratch. In the resource theory of imaginarity, the target resource corresponds to the imaginary part of quantum states, and hence the free states are real states. When the density operator $\rho$ is a real matrix in the given basis $\{\ket{j}\}_{j=0}^{d-1}$, we call $\rho$ a real state (free state). The set of all real states is defined as
\begin{align}
	\mathscr{R} = \{\rho \in \mathscr{D}(\mathscr{H}): \bra{i}\rho\ket{j} \in \mathbb{R}, \; i,j=0,1,\dots,d-1\}. 
\end{align}
From the above definition, it follows that $\rho \in \mathscr{R}$ if and only if $\rho$ is a real symmetric matrix, i.e., $\rho = \rho^{T}$.

In the resource theory of imaginarity, the free operations are incapable of creating imaginarity ex nihilo. Only transformation and redistribution of pre‑existing imaginarity resources are permitted, and no state possessing imaginary parts can be derived from a purely real state by means of free operations, and hence the free operations are real operations. If the Kraus operators $\{K_{m}\}$ of a quantum operation $\Lambda$ (i.e., $\Lambda(\cdot) = \sum_{m} K_{m} (\cdot) K_{m}^{\dagger}$) satisfy $\sum_{m} K_{m}^{\dagger}K_{m} = I$ and $\bra{i}K_{m}\ket{j} \in \mathbb{R}$ for all $m, i, j$, then $\Lambda$ is called a real operation (free operation). This definition is natural, meaning that any real operation acting on any real state will not produce an imaginary state. And the existence of the operator-sum representation is guaranteed by the following theorem \cite{Nielsen10}.
\begin{theorem}
	(Existence of Operator-Sum Representation) All quantum operations $\varepsilon$ on a Hilbert space of dimension $d$ can be generated by an operator-sum representation with at most $d^{2}$ elements,
\begin{align}
		\varepsilon(\rho) \equiv \sum_{k=0}^{M-1} E_{k}\rho E_{k}^{\dagger}, 
\end{align}
	where $1 \le M \le d^{2}$, each $E_{k}$ is a bounded linear operator acting on the d-dimensional Hilbert space and satisfies $\sum_{k} E_{k}^{\dagger}E_{k} \le I$, the equality holds if and only if the quantum operation $\varepsilon$ is trace-preserving. 
\end{theorem}

There is a very important state in the theory of imaginarity called the maximal imaginary state:
\begin{align}
	\ket{+} = \frac{1}{\sqrt{2}}(\ket{0} + i\ket{1}), 
\end{align}
which can be converted to any quantum state of any dimension through real operations. The state $\ket{-} = \frac{1}{\sqrt{2}}(\ket{0} - i\ket{1})$ is also a maximal imaginary state \cite{Hickey18}.

In the framework of the theory of imaginarity, an imaginarity measure $\mathscr{F}: \mathscr{D}(\mathscr{H}) \to [0, +\infty)$ must satisfy the following conditions (1) and (2) \cite{Hickey18}:

(1) Non-negativity: $\mathscr{F}(\rho) \ge 0$, and $\mathscr{F}(\rho) = 0$ if and only if $\rho \in \mathscr{R}$;

(2) Monotonicity: Let $\Lambda$ be a real operation. Then $\mathscr{F}(\Lambda(\rho)) \le \mathscr{F}(\rho)$.

Additionally, some other properties can be discussed:

(3) Strong monotonicity: $\mathscr{F}(\rho) \ge \sum_{j} p_{j}\mathscr{F}(\rho_{j})$, where $p_{j} = \mathrm{Tr}[K_{j}\rho K_{j}^{\dagger}]$, $\rho_{j} = K_{j}\rho K_{j}^{\dagger}/p_j$, and $K_{j}$ are real Kraus operators;

(4) Convexity: $\sum_{j} p_{j}\mathscr{F}(\rho_{j}) \ge \mathscr{F}(\sum_{j} p_{j}\rho_{j})$;

(5) Block additivity: $\mathscr{F}(p\rho_{1} \oplus (1-p)\rho_{2}) = p\mathscr{F}(\rho_{1}) + (1-p)\mathscr{F}(\rho_{2})$.

We will now introduce several well-defined imaginarity measures.

(I) The $l_{1}$-norm-based imaginarity measure \cite{Chen23} can be expressed as
\begin{align}
	 \mathscr{F}_{l_{1}}(\rho) = \min_{\sigma \in \mathscr{R}} \|\rho - \sigma\|_{l_{1}} = \sum_{i \ne j} |\mathrm{Im}(\rho_{ij})|,
\end{align}
where $\rho$ is any quantum state, $\sigma$ is a real state, and $\mathrm{Im}(\rho_{ij})$ represents the imaginary part of the matrix element $\rho_{ij}$.

(II) In quantum resource theories, The robustness of imaginarity is the minimum amount of arbitrary-state admixture required to transform a state into the free-state set. The robustness of imaginarity \cite{Wu21} is defined as
\begin{align}
	\mathscr{F}_{R}(\rho) = \min_{\tau} \left\{s \ge 0: \frac{\rho + s\tau}{1+s} \in \mathscr{R}, \tau \in \mathscr{D}(\mathscr{H})\right\} = \frac{1}{2}\|\rho - \rho^{T}\|_{1},
\end{align}
where $\mathscr{D}(\mathscr{H})$ denotes the set of density operators on the Hilbert space $\mathscr{H}$, $T$ denotes the transpose, and $\|A\|_{1} = \mathrm{Tr}\sqrt{A^{\dagger}A}$ is the trace norm.

(III) The weight-based imaginarity measure \cite{Bu18} is defined as
\begin{align}
	\mathscr{F}_{w}(\rho) = \min_{\sigma, \tau} \left\{s \ge 0: \rho = (1-s)\sigma + s\tau, \sigma \in \mathscr{R}, \tau \in \mathscr{D}(\mathscr{H})\right\}.
\end{align}
Note that in convex resource theories, the weight-based quantifier is closely related to tasks involving exclusion.

(IV) Based on quantum relative entropy, we define the relative entropy of imaginarity for a quantum state $\rho$ as \cite{Xue21}:
\begin{align}\label{the relative entropy of imaginarity}
	 \mathscr{F}_{r}(\rho) = \min_{\sigma \in \mathscr{R}} S(\rho \| \sigma) = S(\mathrm{Re}(\rho)) - S(\rho),
\end{align}
where $S(\rho \| \sigma) = -S(\rho) - \mathrm{Tr}(\rho\log{\sigma}) = \mathrm{Tr}(\rho\log{\rho}) - \mathrm{Tr}(\rho\log{\sigma})$ is the quantum relative entropy, $S(\rho) = -\mathrm{Tr}(\rho\log{\rho})$ is the von Neumann entropy, and $\mathrm{Re}(\rho) = \frac{1}{2}(\rho + \rho^{T})$. All logarithms are taken to base $2$.

\begin{remark}
Within the resource theory of imaginarity, free states are defined as quantum states whose density matrices are real matrices under a fixed orthonormal basis $\{\ket{j}\}_{j=0}^{d-1}$: \begin{align} \rho \in \mathscr{R} \Longleftrightarrow \rho = \rho^{T}. \end{align}
	Here, the transpose operation $T$, the imaginary part of matrix entries \(\mathrm{Im}(\rho_{ij})\), and the concept of real state are all dependent on the selected basis. In other words, a state that qualifies as a real state under one basis may no longer be a real state under another unitarily transformed basis.

For example, we take the single-qubit state: $\ket{+} = \frac{1}{\sqrt{2}}(\ket{0} + i\ket{1})$. This state is a maximal imaginary state under the computational basis \(\{\ket{0},\ket{1}\}\). However, if we adopt a new basis: \begin{align} \ket{e_{0}}=\frac{1}{\sqrt{2}}(\ket{0} + i\ket{1}),\;\;\; \ket{e_{1}}=\frac{1}{\sqrt{2}}(\ket{0} - i\ket{1}) \end{align} the identical physical state $\ket{+}$ corresponds to basis vector \(|e_0\rangle\) in the new basis, with its density matrix written as:$\begin{bmatrix}1&0 \\0&0\end{bmatrix}$ which carries no imaginarity at all. This example directly illustrates that imaginarity is not an invariant under general unitary transformations, but a resource defined relative to a fixed reference basis.

Based on the above discussion, all free states, free operations and imaginarity measures in this paper are defined with respect to a pre-fixed computational basis. For real orthogonal transformations $O$ (i.e. all matrix elements of $O$ are real), real states map to real states (since $\rho=\rho^T \implies O\rho O^T=(O\rho O^T)^T$), and imaginarity measures remain invariant or preserve their structure. For general complex unitary transformations U, the set of real states transforms to $U\rho U^\dagger$. The original structure of free states and free operations will be altered, so the numerical values of imaginarity and its decay formulas will generally change accordingly. Therefore, the results of this paper are not basis-independent, but defined relative to physically meaningful reference bases, such as the computational basis, energy eigenbasis, path basis or polarization basis.
\end{remark}

\section{Decay of Imaginarity of Quantum States After Passing Through Different Quantum Channels}

Similar to quantum entanglement and coherence, quantum imaginarity can also be destroyed as the quantum state evolves. In this section, we mainly explore the changes in different imaginarity measures of a pure state after it passes through a quantum noise channel. 

Here we use the decay of an imaginarity measure to quantify the effect of a quantum noise channel. Let $\varepsilon$ be a quantum channel. The decay of imaginarity is defined as the difference between the initial imaginarity and the imaginarity after passing through the quantum noise channel:
\begin{align}
	\Delta I = \mathscr{F}(\rho) - \mathscr{F}(\varepsilon(\rho)), 
\end{align}
where $\rho$ is a pure state and $\mathscr{F}$ is an imaginarity measure.
All noise channels considered in this section are real operations in the fixed computational basis. Hence, by monotonicity under free operations, $\mathscr F(\varepsilon(\rho))\leq\mathscr F(\rho)$, and therefore $\Delta I\geq0$.

\subsection{Single-Qubit Case}\label{sec:single-qubit-case}

For any pure state $\ket{\psi}$, there exists a real orthogonal matrix $O$ such that 
\begin{align}
	O\ket{\psi }=\sqrt{\frac{1+\left | \left \langle \psi ^{\ast }  | \psi  \right \rangle  \right | }{2} }\ket{0}+i\sqrt{\frac{1-\left | \left \langle \psi ^{\ast }  | \psi  \right \rangle  \right | }{2} }\ket{1},
\end{align}
which keeps the imaginarity unchanged \cite{Wu21}, where $\ket{\psi^{\ast}}$ represents the complex conjugate of $\ket{\psi}$.

We can consider a canonical representative of a pure state under a real orthogonal transformation:
\begin{align}\label{a pure state under a real orthogonal transformation}
	 \ket{\gamma} = \sqrt{\frac{1 + |\left \langle \psi^{\ast} | \psi \right \rangle|}{2}}\ket{0} + i\sqrt{\frac{1 - |\left \langle \psi^{\ast} | \psi \right \rangle|}{2}}\ket{1} = \sqrt{\frac{1+A}{2}}\ket{0} + i\sqrt{\frac{1-A}{2}}\ket{1},
\end{align}
where $A = |\left \langle \psi^{\ast} | \psi \right \rangle|$. The density matrix corresponding to $\ket{\gamma}$ is
\begin{align}
	\ket{\gamma}\bra{\gamma} = \begin{bmatrix}
		\frac{1+A}{2} & -\frac{\sqrt{1-A^{2}}}{2}i \\
		\frac{\sqrt{1-A^{2}}}{2}i & \frac{1-A}{2}
	\end{bmatrix}. 
\end{align}

Consider the dephasing ($D$), generalized amplitude damping ($GAD$) \cite{Nielsen10} and phase-amplitude damping ($PAD$) channels given by the following Kraus operators:
\begin{equation}
\begin{aligned}
	E_{0}^{D}&=\begin{bmatrix}
		\sqrt{1-\frac{p}{2}}   & 0\\
		0 &\sqrt{1-\frac{p}{2}} 
	\end{bmatrix},\qquad
	E_{1}^{D}=\begin{bmatrix}
		\sqrt{\frac{p}{2}}   & 0\\
		0 &-\sqrt{\frac{p}{2}} 
	\end{bmatrix};
\end{aligned}
\end{equation}
\begin{equation}
	\begin{aligned}
		E_{0}^{GAD} &= \sqrt{p_{1}}\begin{bmatrix}
			1 & 0 \\
			0 & \sqrt{1-p_{2}}
		\end{bmatrix} ,
	\qquad	E_{1}^{GAD} = \sqrt{p_{1}}\begin{bmatrix}
			0 & \sqrt{p_{2}} \\
			0 & 0
		\end{bmatrix} ,\\
		E_{2}^{GAD}& = \sqrt{1-p_{1}}\begin{bmatrix}
			\sqrt{1-p_{2}} & 0 \\
			0 & 1
		\end{bmatrix} ,
		E_{3}^{GAD} = \sqrt{1-p_{1}}\begin{bmatrix}
			0 & 0 \\
			\sqrt{p_{2}} & 0
		\end{bmatrix} ;
	\end{aligned}
\end{equation}
\begin{equation}
	\begin{aligned}
		E_{0}^{PAD} &= \begin{bmatrix}
			1 & 0 \\
			0 & \sqrt{1-p_{1}-p_{2}}
		\end{bmatrix} , 
		E_{1}^{PAD} = \begin{bmatrix}
			0 & \sqrt{p_{2}} \\
			0 & 0
		\end{bmatrix} ,
		E_{2}^{PAD} =\begin{bmatrix}
			0 & 0 \\
			0 & \sqrt{p_{1}}
		\end{bmatrix}.
	\end{aligned}
\end{equation}
Under the $D$, $GAD$, $PAD$ channels, the pure state $\ket{\gamma}\bra{\gamma}$ is transformed
into the following forms:
\begin{equation}
	\begin{aligned}
			\varepsilon_{D}(\ket{\gamma}\bra{\gamma}) &= \begin{bmatrix}
			\frac{1+A}{2} & -\frac{(1-p)\sqrt{1-A^{2}}}{2}i \\
			\frac{(1-p)\sqrt{1-A^{2}}}{2}i & \frac{1-A}{2}
		\end{bmatrix},\\
		\varepsilon_{GAD}(\ket{\gamma}\bra{\gamma}) &= \begin{bmatrix}
			\frac{1 + p_{2}(2p_{1}-1) + A(1-p_{2})}{2} & -\frac{\sqrt{(1-p_{2})(1-A^{2})}}{2}i \\
			\frac{\sqrt{(1-p_{2})(1-A^{2})}}{2}i & \frac{1 - p_{2}(2p_{1}-1) - A(1-p_{2})}{2}
		\end{bmatrix},\\
		\varepsilon_{PAD}(\ket{\gamma}\bra{\gamma}) &= \begin{bmatrix}
			\frac{1+A+p_{2}(1-A)}{2} & -\frac{\sqrt{(1-p_{1}-p_{2})(1-A^{2})}}{2}i \\
			\frac{\sqrt{(1-p_{1}-p_{2})(1-A^{2})}}{2}i & \frac{(1-p_{2})(1-A)}{2}
		\end{bmatrix}
	\end{aligned}
\end{equation}
Of particular note, when $p_{1}=1$, the generalized amplitude damping channel degenerates to the standard amplitude damping channel. In the phase-amplitude damping channel, $p_{1}$ is the probability parameter of phase damping, representing the probability of phase loss; $p_{2}$ is the parameter of amplitude damping, representing the probability of the decay from the excited state to the ground state. When $p_{1}=0$, the phase-amplitude damping channel degenerates into an amplitude damping channel, and when $p_{2}=0$, it degenerates into a phase damping channel.

\begin{remark}
Eq.~\eqref{a pure state under a real orthogonal transformation} identifies a canonical representative of a pure state under a real orthogonal transformation. Since both $O$ and $O^\mathrm{T}$ are real operations, an imaginarity measure takes the same value on $\rho$ and $O\rho O^\mathrm{T}$. This invariance, however, concerns the input state only. The quantum channels considered below are defined with respect to the fixed computational basis and are not, in general, covariant under an arbitrary real orthogonal transformation. Consequently, $\mathscr F(\varepsilon(O\rho O^{T}))$ need not be equal to $\mathscr F(\varepsilon(\rho))$. Therefore, in order to analyze arbitrary pure input states, an additional parameter describing the orientation of the state relative to the preferred basis of the channel must be retained.
\end{remark}

As established below, the canonical state in Eq.~\eqref{a pure state under a real orthogonal transformation} can be used to calculate the decays of $\mathscr F_{l_{1}}$ and $\mathscr F_{R}$ for an arbitrary single-qubit input state, even though the channel is not required to be covariant under real orthogonal transformations.

\begin{theorem} \label{thm:real-qubit-attenuation} Let $\varepsilon(\cdot ) = \sum_{m} K_{m} (\cdot)  K_{m}^{\dagger}$ be a real single-qubit quantum channel, where $K_{m}\in\mathbb{R}^{2\times 2}$ and $\sum_{m} K_{m}^{\dagger}K_{m} = I_{2}$. Define $\eta_{\varepsilon} := \sum_m\det K_m$. Then, for every single-qubit density operator $\rho$, we have $\mathscr F_{l_{1}}(\rho)=\mathscr F_{R}(\rho)$, and
		\begin{align}\mathscr F_{l_{1}}\!\left(\varepsilon(\rho)\right) = |\eta_{\varepsilon}|\mathscr F_{l_{1}}(\rho), \qquad  \mathscr F_{R}\!\left(\varepsilon(\rho)\right) = |\eta_{\varepsilon}|\mathscr F_{R}(\rho).    \end{align}
		Moreover, $|\eta_{\varepsilon}|\leq 1$. Consequently, \begin{align} \Delta I_{l_{1}}^{\varepsilon}(\rho) = \Delta I_R^{\varepsilon}(\rho) = \left(1-|\eta_{\varepsilon}|\right)\mathscr F_{l_{1}}(\rho). \end{align} 
		In particular, if $\rho=|\psi\rangle\langle\psi|$ is a pure state and $A=|\langle\psi^{\ast}|\psi\rangle|$,  then \begin{align} \mathscr F_{l_{1}}(\rho)=\mathscr F_R(\rho)=\sqrt{1-A^{2}}, \end{align} and hence \begin{align} \Delta I_{l_{1}}^{\varepsilon} = \Delta I_R^{\varepsilon} = \left(1-|\eta_{\varepsilon}|\right)\sqrt{1-A^{2}}. \end{align} 
\end{theorem}
The detailed proof of Theorem~\ref{thm:real-qubit-attenuation} is provided in Appendix~\ref{app:proof-real-qubit-attenuation}. The theorem shows that, for a single-qubit system, the $\ell_1$-norm-based imaginarity and the robustness of imaginarity are not only equal for the input state, but are also attenuated by exactly the same channel-dependent factor under any real quantum channel. This provides a direct comparison between these two imaginarity measures 
and considerably simplifies their decay analysis.

The exact attenuation theorem immediately yields the following corollary, which justifies the use of the canonical representative in Eq.~\eqref{a pure state under a real orthogonal transformation} when evaluating the decays of $\mathscr F_{l_{1}}$ and $\mathscr F_{R}$.

\begin{corollary} \label{cor:canonical-state-reduction} Let $\varepsilon$ be a real single-qubit quantum channel, and let $O\in O(2)$ be a real orthogonal matrix. Then, for every single-qubit state $\rho$ and for $K\in\{l_1,R\}$, $\mathscr F_K\!\left(\varepsilon(O\rho O^T)\right) = \mathscr F_K\!\left(\varepsilon(\rho)\right)$. Consequently, $\Delta I_K^{\varepsilon}(O\rho O^T) = \Delta I_K^{\varepsilon}(\rho)$, where $K\in\{l_1,R\}$. 
\end{corollary}
The proof of Corollary~\ref{cor:canonical-state-reduction} is given in Appendix~\ref{app:proof-canonical-state-reduction}. Therefore, for the $l_1$-norm-based imaginarity and the robustness of imaginarity, the canonical state introduced in Eq.~\eqref{a pure state under a real orthogonal transformation} can be used without loss of generality when evaluating the decay of an arbitrary single-qubit pure state under a real quantum channel.

In particular, let $\rho_{\psi}=|\psi\rangle\langle\psi|$  be an arbitrary single-qubit pure state, and let $\rho_{\gamma} = |\gamma\rangle\langle\gamma| = O\rho_{\psi}O^T$ be its canonical representative in Eq.~\eqref{a pure state under a real orthogonal transformation}. Corollary \ref{cor:canonical-state-reduction} implies that $\mathscr F_K\!\left(\varepsilon(\rho_{\gamma})\right) = \mathscr F_K\!\left(\varepsilon(\rho_{\psi})\right)$ and $\Delta I_K^{\varepsilon}(\rho_{\gamma}) = \Delta I_K^{\varepsilon}(\rho_{\psi})$, where $K\in\{l_1,R\}$. Therefore, the canonical state in Eq.~\eqref{a pure state under a real orthogonal transformation} can be used without loss of generality to evaluate the decays of $\mathscr F_{l_1}$ and $\mathscr F_R$ for an arbitrary single-qubit pure input state. Notice that this result does not require the covariance relation  $\varepsilon(O\rho O^T) = O\varepsilon(\rho)O^T$. In general, the two output states $\varepsilon(O\rho O^T)$ and $\varepsilon(\rho)$ need not be equal or related by the same orthogonal transformation. The corollary states only that their $\mathscr F_{l_1}$ and $\mathscr F_R$ values, and consequently their corresponding imaginarity decays, are equal.

Since the $D$, $GAD$, $PAD$ channels are real single-qubit quantum channels, Theorem~\ref{thm:real-qubit-attenuation} and Corollary~\ref{cor:canonical-state-reduction} apply directly. Substituting the corresponding attenuation coefficients into the
general decay formula, we obtain
\begin{equation}
		\begin{aligned}
			&	\Delta I_{l_{1} }^{D}=\Delta I_{R }^{D}=p\sqrt{1-A^{2} },\\
			&	\Delta I_{l_{1} }^{GAD}= \Delta I_{R }^{GAD}=(1-\sqrt{1-p_{2} } ) \sqrt{1-A^{2} },\\
			&	\Delta I_{l_{1} }^{PAD}=\Delta I_{R }^{PAD}=(1-\sqrt{1-p_{1}-p_{2} } ) \sqrt{1-A^{2} },
		\end{aligned}
\end{equation}
These expressions hold for an arbitrary single-qubit pure input state and are independent of the orientation of the state within its real-orthogonal orbit.

We next consider the relative entropy of imaginarity $\mathscr{F}_r$. The canonical-state reduction established above is specific to $\mathscr{F}_{\ell_1}$ and $\mathscr{F}_R$. Although a real orthogonal transformation preserves the input relative entropy of imaginarity, $\mathscr{F}_r(O\rho O^{\mathrm T})=\mathscr{F}_r(\rho)$, the equality $\mathscr{F}_r\!\left(\varepsilon(O\rho O^{\mathrm T})\right)=\mathscr{F}_r\!\left(\varepsilon(\rho)\right)$ does not hold in general. Consequently, the parameter $A=|\langle\psi^*|\psi\rangle|$ determines the initial relative entropy of imaginarity of a pure state, but it does not completely determine its decay under a fixed basis-dependent channel. To retain a description of arbitrary pure input states, we introduce an additional orientation parameter $\vartheta$.

To parameterize the entire pure-state Bloch sphere, we first retain the sign of the $y$-component. An arbitrary single-qubit pure state can be written as \begin{equation} \label{eq:general-pure-qubit-chi} \rho_{\chi}(A,\vartheta) = \frac{1}{2} \begin{pmatrix} 1+A\cos\vartheta & A\sin\vartheta-i\chi\sqrt{1-A^2}\\[1mm] A\sin\vartheta+i\chi\sqrt{1-A^2} & 1-A\cos\vartheta \end{pmatrix}, \end{equation} 
where $A\in[0,1], \vartheta\in[0,2\pi), \chi\in\{+1,-1\}$. The corresponding Bloch vector is \begin{align} \boldsymbol r_{\chi}(A,\vartheta) = \left( A\sin\vartheta,\, \chi\sqrt{1-A^2},\, A\cos\vartheta \right). \end{align} Its squared length is \begin{align}  \left\| \boldsymbol r_{\chi}(A,\vartheta) \right\|^2 = A^2\sin^2\vartheta + (1-A^2) + A^2\cos^2\vartheta=1,  \end{align} and hence Eq.~\eqref{eq:general-pure-qubit-chi} represents a pure state.

The two choices of $\chi$ are related by transposition: $\rho_{-\chi}(A,\vartheta) = \rho_{\chi}(A,\vartheta)^{\mathrm T}$. For every real quantum channel  $\varepsilon(\rho) = \sum_m K_m\rho K_m^{\mathrm T}$, where $K_m\in\mathbb R^{2\times2}$,  one has $\varepsilon(\rho^{\mathrm T}) = \varepsilon(\rho)^{\mathrm T}$.  Moreover, the relative entropy of imaginarity is invariant under transposition. Indeed, $\operatorname{Re}(\sigma^{\mathrm T}) = \operatorname{Re}(\sigma) $ and $\sigma^{\mathrm T}$ has the same eigenvalues as $\sigma$, so that \begin{align} \mathscr{F}_r(\sigma^{\mathrm T}) = S\!\left( \operatorname{Re}(\sigma^{\mathrm T}) \right) - S(\sigma^{\mathrm T}) = S\!\left( \operatorname{Re}(\sigma) \right) - S(\sigma) = \mathscr{F}_r(\sigma). \end{align}  It follows that \begin{align} \mathscr{F}_r\!\left( \varepsilon(\rho_{-\chi}(A,\vartheta)) \right)= \mathscr{F}_r\!\left( \varepsilon(\rho_{\chi}(A,\vartheta)^{\mathrm T}) \right)= \mathscr{F}_r\!\left( \varepsilon(\rho_{\chi}(A,\vartheta))^{\mathrm T} \right)= \mathscr{F}_r\!\left( \varepsilon(\rho_{\chi}(A,\vartheta)) \right). \end{align}  Since the corresponding input quantities are also equal, the two choices of $\chi$ give the same decay: \begin{align} \Delta I_r^\varepsilon \bigl(\rho_{-\chi}(A,\vartheta)\bigr) = \Delta I_r^\varepsilon \bigl(\rho_{\chi}(A,\vartheta)\bigr). \end{align}  We may therefore set $\chi=+1$ without loss of generality and write  $\rho(A,\vartheta) := \rho_{+}(A,\vartheta)$. Therefore, any single-qubit pure state can be described, up to transposition, by
\begin{equation}
	\label{eq:general-pure-qubit}
	\rho(A,\vartheta)
	=
	\frac{1}{2}
	\begin{pmatrix}
		1+A\cos\vartheta
		&
		A\sin\vartheta-i\sqrt{1-A^2}\\[1mm]
		A\sin\vartheta+i\sqrt{1-A^2}
		&
		1-A\cos\vartheta
	\end{pmatrix},
\end{equation}
where $A\in[0,1], \vartheta\in[0,2\pi)$. Its Bloch vector is $\boldsymbol{r}(A,\vartheta)=\left(A\sin\vartheta,\,\sqrt{1-A^2},\,A\cos\vartheta\right)$. Since $A^2\sin^2\vartheta+(1-A^2)+A^2\cos^2\vartheta=1$, the state in Eq.~\eqref{eq:general-pure-qubit} is pure.

To express the relative entropy of imaginarity in a compact form, we define the binary entropy function
\begin{align}\label{h_2(q)}
h_2(q)=-q\log{q}-(1-q)\log{(1-q)},
\end{align}
with the convention $0\log{0}:=0$, and introduce
\begin{align}\label{H_t}
\mathcal{H}(t)=h_2\left(\frac{1+t}{2}\right),\qquad0\leq t\leq1.
\end{align}
For a single-qubit state with Bloch vector $(x,y,z)$, the eigenvalues
of the state are $\frac{1\pm\sqrt{x^2+y^2+z^2}}{2}$, whereas the eigenvalues of its real part are $\frac{1\pm\sqrt{x^2+z^2}}{2}$. It follows from the Eq.~\eqref{the relative entropy of imaginarity}
that
\begin{equation}
	\label{eq:relative-entropy-bloch}
	\mathscr{F}_r(\rho)
	=
	\mathcal{H}\left(\sqrt{x^2+z^2}\right)
	-
	\mathcal{H}\left(\sqrt{x^2+y^2+z^2}\right).
\end{equation}
For the pure state $\rho(A,\vartheta)$, we have $\sqrt{x^2+z^2}=A$ and $\sqrt{x^2+y^2+z^2}=1$. Since $\mathcal{H}(1)=0$, the initial relative entropy of imaginarity is $\mathscr{F}_r\!\left(\rho(A,\vartheta)\right)=\mathcal{H}(A)$, which is independent of $\vartheta$. Nevertheless, the relative entropy of imaginarity of the output state generally depends on $\vartheta$.

The three channels considered above transform the Bloch vector in the
common affine form
\begin{align}
x'=\lambda_{\varepsilon}x,\qquad y'=\lambda_{\varepsilon}y,\qquad z'=\mu_{\varepsilon}z+\nu_{\varepsilon}
\end{align}
where
\begin{equation}\label{eq:channel-bloch-parameters}
	\begin{array}{c|ccc}
		\varepsilon
		&
		\lambda_{\varepsilon}
		&
		\mu_{\varepsilon}
		&
		\nu_{\varepsilon}
		\\
		\hline
		D
		&
		1-p
		&
		1
		&
		0
		\\[1mm]
     	GAD
		&
		\sqrt{1-p_2}
		&
		1-p_2
		&
		p_2(2p_1-1)
		\\[1mm]
		PAD
		&
		\sqrt{1-p_1-p_2}
		&
		1-p_2
		&
		p_2
	\end{array}
\end{equation}
For the $\mathrm{PAD}$ channel, the parameters satisfy $p_1\geq0, p_2\geq0, p_1+p_2\leq1$. For the pure input state $\rho(A,\vartheta)$, define
\begin{equation}
	u_{\varepsilon}(A,\vartheta)=\sqrt{\lambda_{\varepsilon}^{2}A^{2}\sin^{2}\vartheta+\left(\mu_{\varepsilon}A\cos\vartheta+\nu_{\varepsilon}\right)^{2}},
\end{equation}
\begin{equation}
	v_{\varepsilon}(A,\vartheta)
	=
	\sqrt{
		u_{\varepsilon}(A,\vartheta)^{2}
		+
		\lambda_{\varepsilon}^{2}(1-A^{2})
	}.
\end{equation}
Here, $u_{\varepsilon}(A,\vartheta)$ is the length of the Bloch vector associated with the real part of the output state, whereas $v_{\varepsilon}(A,\vartheta)$ is the length of the full output Bloch vector. Therefore,
\begin{equation}
	\mathscr{F}_r\!\left(\varepsilon\!\left(\rho(A,\vartheta)\right)\right)=\mathcal{H}\!\left(u_{\varepsilon}(A,\vartheta)\right)-\mathcal{H}\!\left(v_{\varepsilon}(A,\vartheta)\right).
\end{equation}
Consequently, the decay of the relative entropy of imaginarity for an arbitrary single-qubit pure input state is
\begin{equation}\label{eq:general-relative-entropy-decay}
\Delta I_r^{\varepsilon}(A,\vartheta)=\mathcal{H}(A)-\mathcal{H}\!\left(u_{\varepsilon}(A,\vartheta)\right)+\mathcal{H}\!\left(v_{\varepsilon}(A,\vartheta)\right).
\end{equation}
The canonical state in Eq.~\eqref{a pure state under a real orthogonal transformation} corresponds to $\vartheta=0$. In this case,
\begin{align}
u_{\varepsilon}(A,0)=\left|\mu_{\varepsilon}A+\nu_{\varepsilon}\right|,\quad v_{\varepsilon}(A,0)=\sqrt{\left(\mu_{\varepsilon}A+\nu_{\varepsilon}\right)^{2}+\lambda_{\varepsilon}^{2}(1-A^{2})}.
\end{align}
Therefore, by setting $\vartheta=0$ in Eq.~\eqref{eq:general-relative-entropy-decay}, substituting the parameters of the $D$, $GAD$, and $PAD$ channels, we obtain the following explicit expressions:
\begin{equation}
\begin{aligned}
			\Delta I_{r}^{D}
			=\mathcal H(\sqrt{A^2+(1-p)^2(1-A^2)}),
\end{aligned}
\end{equation}			
\begin{equation}
\begin{aligned}
			\Delta I_{r}^{GAD}=\mathcal H(A)+\mathcal H\!\left(\sqrt{\alpha }\right)
			-\mathcal H(\left | (1-p_2)A+p_2(2p_1-1) \right |),
\end{aligned}
\end{equation}
\begin{equation}
\begin{aligned}			
			\Delta I_{r}^{PAD}=\mathcal H(A)+\mathcal H(\sqrt{\beta})-\mathcal H(p_2+(1-p_2)A),
\end{aligned}
\end{equation}
where $\alpha=(p_{2}{-}2p_{1}p_{2}{-}(1{-}p_{2})A)^{2} {+}(1{-}p_{2})(1{-}A^{2} )$, $\beta=1{-}p_{2}(1{-}p_{2})(1{-}A)^{2}{-}p_{1}(1{-}A^{2} )$, and the function $\mathcal{H}(\cdot)$ is given by Eq.~\eqref{H_t}.

\begin{figure}[H]
	\centering
	\includegraphics[width=0.9\linewidth]{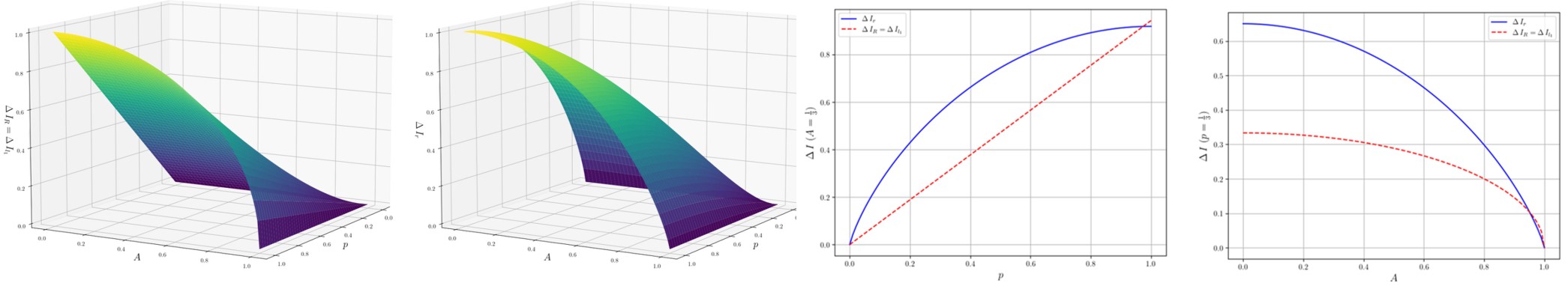}
	\caption{Decay of $\mathscr{F}_{l_{1}}$, $\mathscr{F}_{R}$, and $\mathscr{F}_{r}$ under a Dephasing Channel}
	\label{fig:1}
\end{figure}
\begin{figure}[H]
	\centering
	\includegraphics[width=0.8\linewidth]{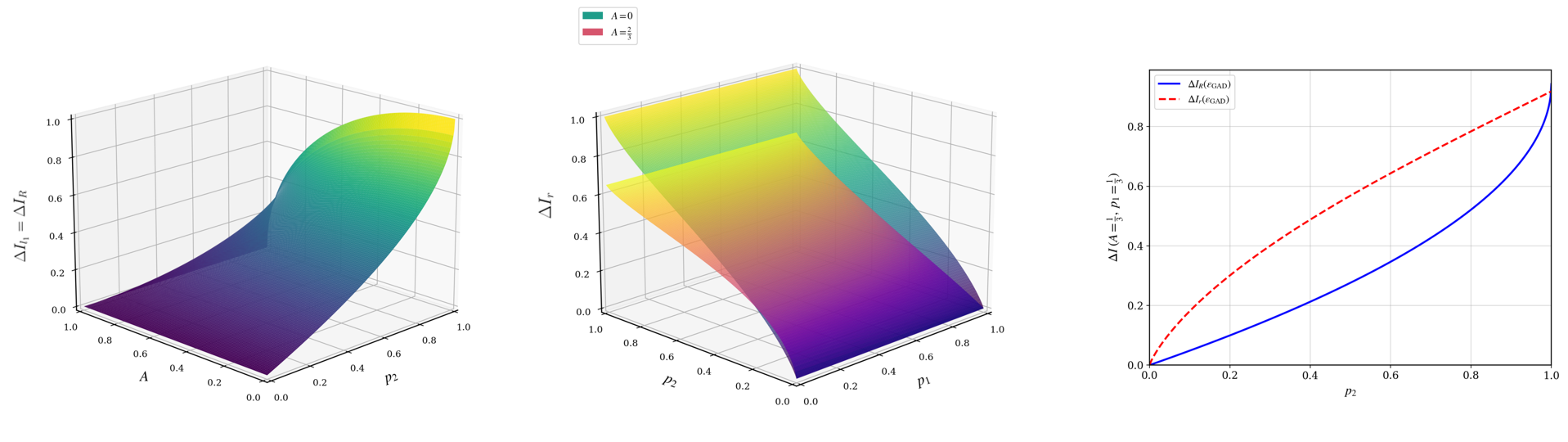}
	\caption{Decay of $\mathscr{F}_{l_{1}}$, $\mathscr{F}_{R}$, and $\mathscr{F}_{r}$ under a Generalized Amplitude Damping Channel}
	\label{fig:2}
\end{figure}
\begin{figure}[H]
	\centering
	\includegraphics[width=0.8\linewidth]{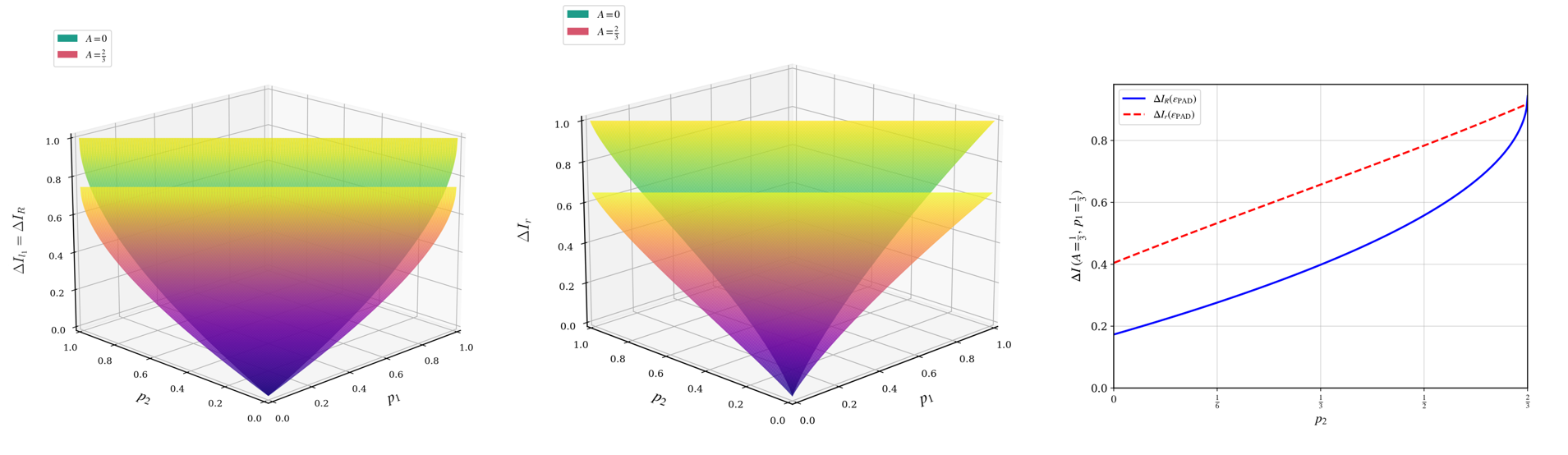}
	\caption{Decay of $\mathscr{F}_{l_{1}}$, $\mathscr{F}_{R}$, and $\mathscr{F}_{r}$ under a Phase-Amplitude Damping Channel}
	\label{fig:3}
\end{figure}

Figure~\ref{fig:1} summarizes the decay of the three imaginarity measures under the dephasing channel for $0\leq p\leq 1$. The $l_1$-norm-based decay and the robustness-based decay coincide and are independent of the real-orthogonal orientation of the input state. For the relative-entropy-based decay, we present the canonical case $\vartheta=0$. For any fixed $A<1$, all three decay quantities vanish at $p=0$ and increase with $p$. In particular, $\Delta I_{l_1}^{D}=\Delta I_R^{D}$ depends linearly on $p$, whereas $\Delta I_r^{D}$ is a concave increasing function of $p$. For any fixed $p>0$, all three quantities decrease concavely with $A$, attaining their maxima at $A=0$ and vanishing at $A=1$. The two cross sections illustrate the different functional dependence of the measures on the channel parameter and the input-state imaginarity.

Figure~\ref{fig:2} summarizes the decay of the three imaginarity measures under the generalized amplitude-damping channel. The $l_1$-norm-based and robustness-based decays coincide. We can see that these two decay quantities are independent of $p_1$, increase convexly with $p_2$ for fixed $A<1$, and decrease concavely with $A$ for fixed $p_2>0$. For the relative-entropy-based decay, the figure shows the canonical case $\vartheta=0$. Unlike the first two measures, $\Delta I_r^{GAD}(A,0)$ generally depends on both $p_1$ and $p_2$. This can be seen directly from the output state: $p_1$ enters its diagonal elements through the term $p_2(2p_1-1)$, whereas the magnitude of the imaginary off-diagonal element depends only on $p_2$. Consequently, changing $p_1$ can affect $\Delta I_r^{GAD}$, but does not affect $\Delta I_{l_1}^{GAD}$ or $\Delta I_R^{GAD}$. At $p_2=0$, the three decay quantities vanish, whereas at $p_2=1$ the output state is real and the decay of each measure equals its initial value. The right panel shows the cross section $A=p_1=1/3$, along which all three decay quantities increase monotonically with $p_2$.

Figure~\ref{fig:3} presents the corresponding results for the phase-amplitude-damping channel over the physical parameter region $p_1\geq0$, $p_2\geq0$, and $p_1+p_2\leq1$. The $l_1$-norm-based and robustness-based decays coincide. Thus, these two quantities depend only on the total damping parameter $p_1+p_2$. For a fixed $A$, all three decays vanish at $(p_1,p_2)=(0,0)$ and attain their maximal values on the boundary $p_1+p_2=1$. The corresponding maxima are $\sqrt{1-A^2}$ for $\Delta I_{l_1}^{\mathrm{PAD}}=\Delta I_R^{\mathrm{PAD}}$ and $\mathcal H(A)$ for $\Delta I_r^{\mathrm{PAD}}$. Unlike the first two measures, the relative-entropy-based decay generally depends on $p_1$ and $p_2$ separately. Along the cross-section $A=p_1=1/3$ shown in the right panel, all three decays increase monotonically with $p_2$. The $l_1$-norm-based and robustness-based decays are convex, whereas the relative-entropy-based decay changes curvature and is therefore neither globally convex nor globally concave over the displayed interval.

\subsection{Two-Qubit Case}

Here we will discuss the decay of imaginarity for different measures of some two-qubit states after passing through a quantum channel. The two state families considered below are supported on two-dimensional subspaces and may therefore be regarded as logical single-qubit states. We use the input-state parametrization and the input imaginarity relations established in Section~\ref{sec:single-qubit-case} in the corresponding logical bases, while treating the action of each physical two-qubit channel separately.

\subsubsection[Decay under a Two-Qubit Bit-Flip Channel]{Decay of Imaginarity of the Pure State $\ket{\gamma}=\alpha\ket{00}+\beta\ket{11}$ under a Two-Qubit Bit-Flip Channel}

For the pure state $\ket{\gamma} = \alpha\ket{00} + \beta\ket{11}$, where $\alpha$ and $\beta$ are complex numbers satisfying $|\alpha|^{2} + |\beta|^{2} = 1$, the state $\ket{\gamma}$ is an entangled state when both $\alpha$ and $\beta$ are non-zero. It can be written as
\begin{align}
	\ket{\gamma} = \cos\frac{\theta}{2}\ket{00} + e^{i\phi}\sin\frac{\theta}{2}\ket{11}, 
\end{align}
where $\theta \in [0, \pi]$ and $\phi \in [0, 2\pi]$. Its corresponding density matrix is
\begin{align}
 \rho=\ket{\gamma } \bra{\gamma}=\frac{1}{2}\begin{bmatrix}\begin{smallmatrix}
 		1+\cos\theta &0  & 0 & e^{-i\phi}\sin\theta\\
 		0 & 0 &0  & 0\\
 		0 & 0 &0  &0 \\
 		e^{i\phi}\sin\theta  &0 &0  &1 - \cos\theta
 \end{smallmatrix}\end{bmatrix}.
\end{align}
For the single-qubit pure state $\ket{\psi} = \cos\frac{\theta}{2}\ket{0} + e^{i\phi}\sin\frac{\theta}{2}\ket{1}$, there exists an orthogonal matrix $O$ such that
\begin{align}
	O\ket{\psi} = \sqrt{\frac{1+A}{2}}\ket{0} + i\sqrt{\frac{1-A}{2}}\ket{1},
\end{align}
where $A = |\left \langle \psi^{\ast} | \psi \right \rangle|$ \cite{Wu21}. Let the orthogonal matrix be $O = \begin{bmatrix} a & b \\ c & d \end{bmatrix}$. Then
\begin{align}
 O\ket{\psi}\bra{\psi}O^{T} = \frac{1}{2} \begin{bmatrix} a & b \\ c & d \end{bmatrix} \begin{bmatrix} 1+\cos\theta & e^{-i\phi}\sin\theta \\ e^{i\phi}\sin\theta & 1-\cos\theta \end{bmatrix} \begin{bmatrix} a & c \\ b & d \end{bmatrix} = \begin{bmatrix} \frac{1+A}{2} & -\frac{\sqrt{1-A^{2}}}{2}i \\ \frac{\sqrt{1-A^{2}}}{2}i & \frac{1-A}{2} \end{bmatrix}. 
\end{align}
Therefore, for the $\rho$ above, there exists an orthogonal matrix $O' = \begin{bmatrix}\begin{smallmatrix} a & 0 & 0 & b \\ 0 & 1 & 0 & 0 \\ 0 & 0 & 1 & 0 \\ c & 0 & 0 & d \end{smallmatrix}\end{bmatrix}$ such that
\begin{align}
	 O'\rho O'^{T} = \begin{bmatrix}\begin{smallmatrix}
	\frac{1+A}{2} & 0 & 0 & -\frac{\sqrt{1-A^{2}}}{2}i \\
	0 & 0 & 0 & 0 \\
	0 & 0 & 0 & 0 \\
	\frac{\sqrt{1-A^{2}}}{2}i & 0 & 0 & \frac{1-A}{2}
\end{smallmatrix}\end{bmatrix}. 
\end{align}

Identifying $\ket{00}$ and $\ket{11}$ as the logical basis states $\ket{0_L}$ and $\ket{1_L}$, respectively, the state in Eq.~\eqref{eq:even-parity-canonical-state} is the canonical orientation $\vartheta=0$ of the pure-state parametrization introduced in Section~\ref{sec:single-qubit-case}. The input-state imaginarity relations derived there therefore apply in this logical subspace.

However, the physical two-qubit bit-flip channel can populate both the even- and odd-parity sectors. Its action on the output imaginarity must therefore be checked separately. As shown below, the decays of $\mathscr F_{l_1}$ and $\mathscr F_R$ are independent of the logical orientation, whereas the relative-entropy expression in Eq.~\eqref{eq:bf-relative-entropy-canonical} refers specifically to the canonical orientation $\vartheta=0$.

For the explicit calculation, we use the canonical state
\begin{equation}
	\label{eq:even-parity-canonical-state}
	\lvert\gamma\rangle=\sqrt{\frac{1+A}{2}}\lvert00\rangle+i\sqrt{\frac{1-A}{2}}\lvert11\rangle.
\end{equation}

The Kraus operators of the two-qubit bit-flip channel are given in Appendix~\ref{app:kraus-operators}; see Eq.~\eqref{eq:BF}. In the parameter convention used there, $p_i$ is the no-flip probability of the $i$th qubit. After passing through the channel, the canonical state becomes
\begin{align}\label{eq:bf-canonical-output} 
	\varepsilon _{BF} (\ket{\gamma}\bra{\gamma})=\begin{bmatrix}\begin{matrix}
			\frac{1-\alpha+(p_{1}+p_{2}-1)A}{2}  &  0&0  &\frac{(1-p_{1}-p_{2} )\sqrt{1-A^{2} }}{2}i   \\
			0&  \frac{\alpha+(p_{1}-p_{2})A}{2} &\frac{(p_{2}-p_{1})\sqrt{1-A^{2} }}{2}i   & 0\\
			0& \frac{(p_{1}-p_{2})\sqrt{1-A^{2} }}{2}i & \frac{\alpha+(p_{2}-p_{1})A}{2} &0 \\
			-\frac{(1-p_{1}-p_{2} )\sqrt{1-A^{2} }}{2}i  &0  &0  &\frac{1-\alpha-(p_{1}+p_{2}-1)A}{2}
	\end{matrix}\end{bmatrix}  ,
\end{align}
where\;\;$\alpha=p_{1}+p_{2}-2p_{1}p_{2}$.

Let $y_L$ denote the $y$-component of the logical Bloch vector. For an arbitrary pure state in the logical subspace $\operatorname{span}\{\lvert00\rangle,\lvert11\rangle\}$, the single-qubit relations established in Section~\ref{sec:single-qubit-case} give $\mathscr F_{l_1}(\rho)=\mathscr F_R(\rho)=|y_L|=\sqrt{1-A^2}$.

Under the two-qubit bit-flip channel, the imaginary off-diagonal components are distributed into the even- and odd-parity blocks with multiplicative factors $p_1+p_2-1$ and $p_1-p_2$, respectively. Since the two blocks have orthogonal supports,
\begin{align}
	\mathscr F_{l_1}(\varepsilon_{\mathrm{BF}}(\rho))=\mathscr F_R(\varepsilon_{\mathrm{BF}}(\rho))=\left(|p_1+p_2-1|+|p_1-p_2|\right)\sqrt{1-A^2}.
\end{align}

Consequently,
\begin{align}
\Delta I_{l_1}^{\mathrm{BF}}=\Delta I_R^{\mathrm{BF}}=\left[1-|p_1+p_2-1|-|p_1-p_2|\right]\sqrt{1-A^2}.
\end{align}

The above expression holds for every pure state in the logical subspace and is independent of the logical orientation parameter
$\vartheta$.

For the relative entropy of imaginarity, we retain the canonical input state in Eq.~\eqref{eq:even-parity-canonical-state}, corresponding to $\vartheta=0$. Direct evaluation of the output state in Eq.~\eqref{eq:bf-canonical-output} gives
\begin{equation}\label{eq:bf-relative-entropy-canonical}
		\begin{aligned}
			\bigtriangleup I_{r}=\mathcal H(A)+h_2(p_1)+h_2(p_2)-h_2(\eta)-(1-\eta)\mathcal H\!\left(\frac{|p_1+p_2-1|A}{1-\eta}\right)-\eta
			\mathcal H\!\left(\frac{|p_1-p_2|A}{\eta}\right)
		\end{aligned}
\end{equation}
where $\eta=p_1+p_2-2p_1p_2$, and the functions $\mathcal{H}(\cdot)$ and $h_2(\cdot)$ are given by Eqs.~\eqref{H_t} and \eqref{h_2(q)}.

Note that Eq.~\eqref{eq:bf-relative-entropy-canonical} is the canonical‑orientation specialization $\vartheta=0$. For a general logical orientation, the relative‑entropy decay may depend on $\vartheta$.

\begin{figure}[H]
	\centering
	\includegraphics[width=0.9\linewidth]{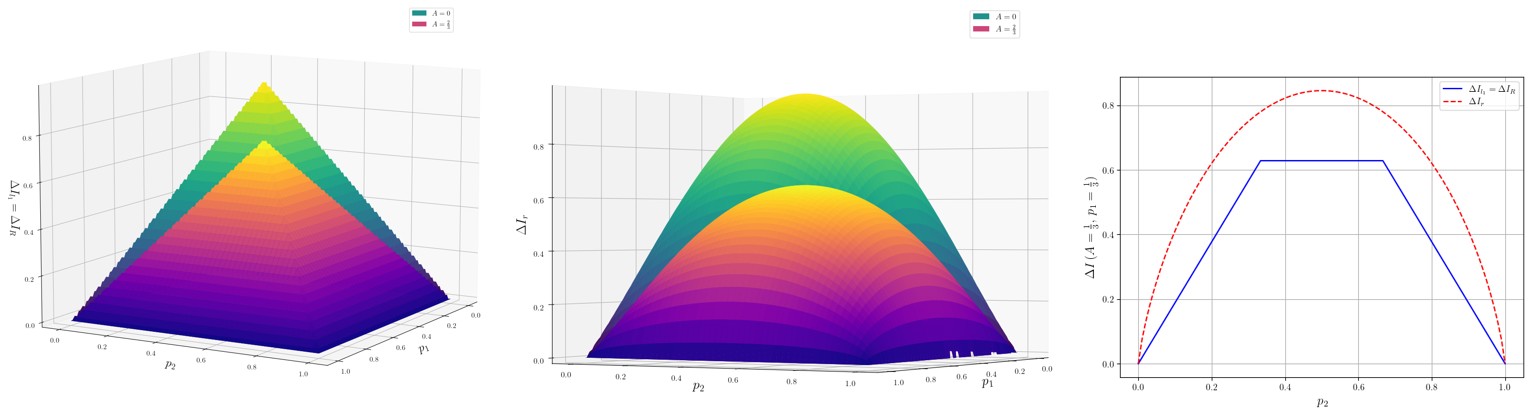}
	\caption{Decay of $\mathscr{F}_{l_{1}}$, $\mathscr{F}_{R}$, and $\mathscr{F}_{r}$ for the pure state $\ket{\gamma} = \alpha\ket{00} + \beta\ket{11}$ under a two-qubit bit-flip channel}
	\label{fig:4}
\end{figure}
Figure~\ref{fig:4} summarizes the imaginarity losses under the two-qubit bit-flip channel, where $p_i$ denotes the no-flip probability of qubit $i$. The left and middle panels show $\Delta I_{\ell_1}^{\mathrm{BF}}=\Delta I_R^{\mathrm{BF}}$ and $\Delta I_r^{\mathrm{BF}}$, respectively, as functions of $(p_1,p_2)$ for $A=0$ and $A=2/3$. The relative-entropy-based results correspond to the canonical orientation $\vartheta=0$. All three loss quantities vanish on the boundary $p_1\in\{0,1\}$ or $p_2\in\{0,1\}$ and attain their global maxima at $p_1=p_2=1/2$. The corresponding maximum values are $\sqrt{1-A^2}$ for $\Delta I_{\ell_1}^{\mathrm{BF}}=\Delta I_R^{\mathrm{BF}}$ and $\mathcal H(A)$ for $\Delta I_r^{\mathrm{BF}}$. Therefore, the maximum losses decrease as $A$ increases. The right panel shows the cross-section $A=p_1=1/3$. On this cross-section, $\Delta I_{\ell_1}^{\mathrm{BF}}=\Delta I_R^{\mathrm{BF}}$ increases linearly on $0\leq p_2\leq1/3$, remains constant on $1/3\leq p_2\leq2/3$, and decreases linearly on $2/3\leq p_2\leq1$. The relative-entropy-based loss is symmetric about $p_2=1/2$ and strictly concave: it increases on $0\leq p_2\leq1/2$, decreases on $1/2\leq p_2\leq1$, and attains its maximum at $p_2=1/2$.

\subsubsection{Decay of Imaginarity of a Dual-Rail Two-Qubit State under an Amplitude Damping Channel}

Suppose a qubit state is represented by two qubits as $\ket{\psi} = \alpha\ket{01} + \beta\ket{10}$, where $\alpha$ and $\beta$ are complex numbers satisfying $|\alpha|^{2} + |\beta|^{2} = 1$. We call this a dual-rail two-qubit state.

Here, we assume that the two rails independently undergo amplitude damping with the same damping probability $\gamma$, so that the overall channel is described by $\varepsilon_{AD} \otimes \varepsilon_{AD}$. So we apply $\varepsilon_{AD} \otimes \varepsilon_{AD}$ to the quantum state $\ket{\psi}$, where $\varepsilon_{AD}$ is the single-qubit amplitude damping channel.
The Kraus operators for the channel $\varepsilon_{AD} \otimes \varepsilon_{AD}$ are given in Appendix~\ref{app:kraus-operators}; see Eq.~\eqref{eq:AD}. Note that the dual‑rail encoding is an error‑detecting encoding rather than, by itself, a complete quantum error‑correcting code.

Following the single-qubit analysis in Section~\ref{sec:single-qubit-case}, the decays of $\mathscr F_{l_1}$ and $\mathscr F_R$ are independent of the logical orientation, and we use the canonical state 
\begin{align}
	\ket{\psi} = \sqrt{\frac{1+A}{2}}\ket{01} + i\sqrt{\frac{1-A}{2}}\ket{10}. 
\end{align}
The state of $\ket{\psi}\bra{\psi}$ after passing through the amplitude damping channel becomes:
\begin{align}
	\begin{bmatrix}\begin{smallmatrix}
			\gamma & 0 & 0 & 0 \\
			0 & \frac{(1-\gamma)(1+A)}{2} & -\frac{(1-\gamma)\sqrt{1-A^{2}}}{2}i & 0 \\
			0 & \frac{(1-\gamma)\sqrt{1-A^{2}}}{2}i & \frac{(1-\gamma)(1-A)}{2} & 0 \\
			0 & 0 & 0 & 0
	\end{smallmatrix}\end{bmatrix}. 
\end{align}
Then we obtain the following decay expressions for $\mathscr{F}_{l_{1}}(\ket{\psi}\bra{\psi})$ and $\mathscr{F}_{R}(\ket{\psi}\bra{\psi})$ :
\begin{align}\label{dual-rail,l1,R}
	\bigtriangleup I_{l_{1} }=\bigtriangleup I_{R }=\gamma \sqrt{1-A^{2} } 
\end{align}
The above expression holds for every pure state in the logical subspace and is independent of the logical orientation parameter
$\vartheta$.

Similar to Section~\ref{sec:single-qubit-case}, we consider general quantum states parameterized by $\vartheta$. We obtain
\begin{align}
	\rho_{\text{out}}(A,\vartheta)=\begin{bmatrix}\begin{smallmatrix}
		\gamma & 0 & 0 & 0\\
		0 & \dfrac{(1-\gamma)(1+A\cos\vartheta)}{2} & \dfrac{(1-\gamma)\left(A\sin\vartheta-i\sqrt{1-A^2}\right)}{2} & 0\\
		0 & \dfrac{(1-\gamma)\left(A\sin\vartheta+i\sqrt{1-A^2}\right)}{2} & \dfrac{(1-\gamma)(1-A\cos\vartheta)}{2} & 0\\
		0 & 0 & 0 & 0
\end{smallmatrix}\end{bmatrix}
\end{align}
By explicit computation, we obtain the following decay expression for $\mathscr{F}_{r}(\ket{\psi}\bra{\psi})$:
\begin{align}\label{dual-rail,r}
	\bigtriangleup I_{r}=\gamma \mathcal H(A)
\end{align}
We can see that the decay of $\mathscr{F}_{r}$ is independent of $\vartheta$.
\begin{figure}[H]
	\centering
	\includegraphics[width=0.9\linewidth]{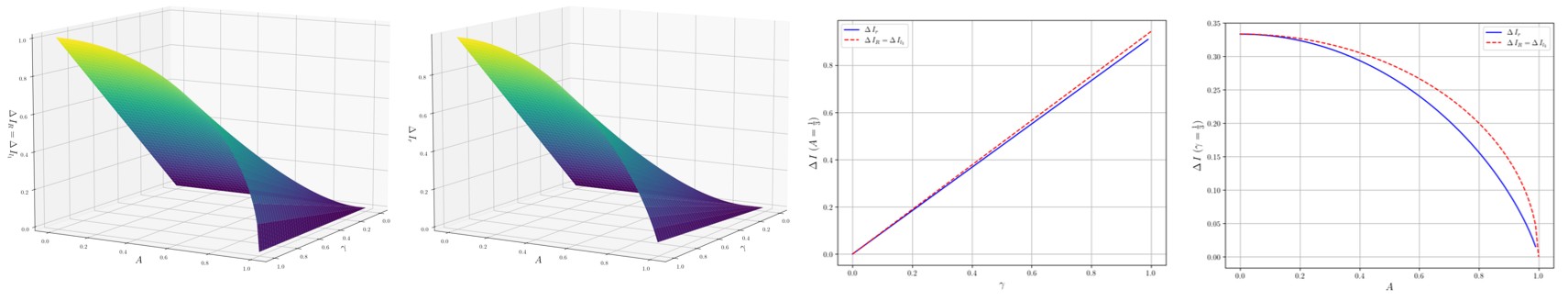}
	\caption{Decay of $\mathscr{F}_{l_{1}}$, $\mathscr{F}_{R}$, and $\mathscr{F}_{r}$ for a dual-rail two-qubit state under an amplitude damping channel}
	\label{fig:5}
\end{figure}
As shown in Figure~\ref{fig:5}, based on the expressions for $\Delta I_{l_1}$, $\Delta I_{R}$, and $\Delta I_{r}$, we plot their respective curves and analyze their variation trends by tuning parameters $A$ and $\gamma$. It can be observed that $\Delta I_{l_1}$, $\Delta I_{R}$, and $\Delta I_{r}$ all attain their maximum values at the maximal imaginary state ($A=0$) and reduce to 0 at the real state ($A=1$). For a fixed parameter $A$ (corresponding to a given quantum state), $\Delta I_{l_1}$, $\Delta I_{R}$, and $\Delta I_{r}$ are monotonically increasing functions of $\gamma$. For a fixed parameter $\gamma$ (corresponding to a given amplitude damping channel), $\Delta I_{l_1}$, $\Delta I_{R}$, and $\Delta I_{r}$ are concave decreasing functions of $A$.

Taken together, Figs.~\ref{fig:1}--\ref{fig:5} show that the $l_1$-norm-based and robustness-based decays coincide for the real single-qubit channels considered here, whereas the relative-entropy-based decay can exhibit a different dependence on the channel and orientation parameters. In the symmetric dual-rail loss model, however, all three measures predict the same fractional loss. Thus, there is no channel- and parameter-independent ordering of the three decay measures.

\subsubsection{Physical Interpretation and Application to a Dual-Rail Photonic Qubit}

To provide a concrete physical interpretation of the above results, we consider a photonic implementation of the dual-rail qubit. The logical basis states are represented by \begin{align}
\ket{0_L} = \ket{01}, \qquad \ket{1_L} = \ket{10}, 
\end{align}
where the two basis states correspond to a single excitation occupying one of two optical modes (rails).

In such a photonic realization, amplitude damping can be used to model single-photon loss. We assume identical and independent loss in the two rails, with $\gamma$ denoting the photon-loss probability. Since a dual-rail logical state contains exactly one excitation, the channel transforms the input state 
$\rho_{\mathrm{in}}$ as \begin{align}
\rho_{\mathrm{out}}
=
(1-\gamma)\rho_{\mathrm{in}}
+
\gamma \ket{00}\bra{00}.
\
\end{align}
The state $\ket{00}$ represents the vacuum state and lies outside the logical code space \begin{align}
\mathcal{C}
=
\operatorname{span}\{\ket{01},\ket{10}\}.
\end{align}
Therefore, a photon-loss event transfers the state outside the code space and can be identified as an erasure event.

For the three imaginarity measures considered in this work, the above 
results in Eqs.~\eqref{dual-rail,l1,R} and \eqref{dual-rail,r} can be expressed in the unified form \begin{align}
\Delta I_K
=
\gamma \mathscr F_K(\rho_{\mathrm{in}}),
\qquad
K\in\{l_1,R,r\},
\end{align}
or equivalently,\begin{align}
	\mathscr F_K(\rho_{\mathrm{out}})
	=
	(1-\gamma)\mathscr F_K(\rho_{\mathrm{in}}).
\end{align}
Therefore, the fraction of imaginarity retained after transmission is $1-\gamma$, while the fraction lost due to photon loss is $\gamma$.

This relation also provides a simple quantitative loss criterion. 
If at least a fraction $q$ of the initial imaginarity is required to remain 
after transmission, then \begin{align}
\mathscr F_K(\rho_{\mathrm{out}})
\geq
q \mathscr F_K(\rho_{\mathrm{in}})
\end{align}
implies \begin{align}
	\gamma \leq 1-q.
\end{align}
Thus, the analytical decay formulas provide a direct constraint on the allowable photon-loss probability for preserving a prescribed fraction of the initial imaginarity resource.

It is worth noting that, conditioned on the no-loss outcome, the normalized state remains exactly $\rho_{\mathrm{in}}$. Hence, the imaginarity of the successfully transmitted dual-rail qubit is fully preserved, while the success probability is $1-\gamma$. The reduction of the unconditional imaginarity therefore originates entirely from the erasure branch.

These results show that the present imaginarity-decay formulas can be used to quantitatively characterize resource degradation caused by photon loss in a dual-rail encoding. We emphasize that dual-rail encoding makes the single-photon-loss event detectable as an erasure; additional redundancy and recovery operations are required for complete quantum error correction.

\section{Operational De-Imaginary Powers and Reference-State Imaginarity Losses of Two-Qubit Channels}\label{sec:two-qubit-channel-analysis}

Quantum channels may generate, preserve, or destroy imaginarity. In this section, we study these abilities for two‑qubit channels acting
on separable states. The channel‑power analysis below is restricted to the robustness of imaginarity $\mathscr{F}_R$ and the relative entropy of imaginarity $\mathscr{F}_r$. Both quantities are treated as resource monotones under real operations in the full two‑qubit Hilbert space.

First, we give the definition of a separable state:
\begin{definition}[Separable state]
	For a composite quantum system consisting of two subsystems $A$ and $B$, with corresponding Hilbert spaces $\mathscr{H}_{A}$ and $\mathscr{H}_{B}$, the Hilbert space for the entire composite system is $\mathscr{H}_{A} \otimes \mathscr{H}_{B}$. If a state $\rho$ in the composite system can be expressed as
\begin{align}
		 \rho = \sum_{i} p_{i}\rho_{A}^{i} \otimes \rho_{B}^{i}, 
\end{align}
where $p_{i} \ge 0$, $\sum_{i} p_{i} = 1$, and $\rho_{A}^{i}$ and $\rho_{B}^{i}$ are states in subsystems $A$ and $B$, respectively, then we call $\rho$ a separable state, denoted by $\rho_{sep}$.
\end{definition}

\subsection{Operational Resource Order on Separable States}

Motivated by the conversion-based characterization of the maximally imaginary state in \cite{Hickey18}, we introduce an analogous operational notion restricted to separable two-qubit states.

Let $\mathrm{Sep}$ denote the set of separable two-qubit states.

\begin{definition}[Real-operation preorder]
For $\rho,\sigma\in\mathrm{Sep}$, we write $\rho\succeq_{\mathbb R}\sigma$ if there exists a real completely positive trace-preserving map $\Lambda$ such that $\sigma=\Lambda(\rho)$. Thus, $\rho\succeq_{\mathbb R}\sigma$ means that $\rho$ can be converted into $\sigma$ by a free real operation.
\end{definition}

\begin{definition}[Operationally maximal separable state]
A separable state $\Omega_{\mathrm{sep}}\in\mathrm{Sep}$ is called operationally maximal with respect to real operations if $\Omega_{\mathrm{sep}}\succeq_{\mathbb R}\rho_{\mathrm{sep}}$ for every $\rho_{\mathrm{sep}}\in\mathrm{Sep}$. Equivalently, for every separable state $\rho_{\mathrm{sep}}$, there exists a real quantum channel $\Lambda$ such that $\rho_{\mathrm{sep}}=\Lambda(\Omega_{\mathrm{sep}})$.
\end{definition}
Let $\omega_{+}=|+\rangle\langle+|$, where $|+\rangle=\frac{|0\rangle+i|1\rangle}{\sqrt{2}}$, and define $\Omega_{++}=\omega_{+}\otimes\omega_{+}$.

\begin{theorem}[Operational maximality of $\Omega_{++}$]
	\label{thm:operational-maximal-separable}
	The state $\Omega_{++}=|+\rangle\langle+|\otimes|+\rangle\langle+|$ is operationally maximal in the set of separable two-qubit states.
\end{theorem}

\begin{proof}
Let $\rho_{\mathrm{sep}}=\sum_ip_i\,\rho_A^{(i)}\otimes\rho_B^{(i)}$ be an arbitrary separable two-qubit state, where $p_i\geq0$, $\sum_i p_i=1$. By the single-qubit conversion property of the maximally imaginary state, for each $i$ there exist real quantum channels $\Lambda_A^{(i)}$ and $\Lambda_B^{(i)}$ such that
\begin{align}
\rho_A^{(i)}=\Lambda_A^{(i)}(\omega_{+}),\qquad \rho_B^{(i)}=\Lambda_B^{(i)}(\omega_{+}).
\end{align}
Let $\left\{K_{A,m}^{(i)}\right\}_{m}$ and $\left\{K_{B,n}^{(i)}\right\}_{n}$ be real Kraus representations of $\Lambda_A^{(i)}$ and $\Lambda_B^{(i)}$, respectively. Define
\begin{align}
L_{i,m,n}=\sqrt{p_i}\,K_{A,m}^{(i)}\otimes K_{B,n}^{(i)}.
\end{align}
	All $L_{i,m,n}$ are real, and
\begin{equation}
\begin{aligned}
\sum_{i,m,n}L_{i,m,n}^{\dagger}L_{i,m,n}
	&=\sum_i p_i\left(\sum_mK_{A,m}^{(i)\dagger}K_{A,m}^{(i)}\right)\otimes\left(\sum_nK_{B,n}^{(i)\dagger}K_{B,n}^{(i)}\right)\\
	&=\sum_i p_i\,I_A\otimes I_B=I_{AB}.
\end{aligned}
\end{equation}
Hence, the operators $\{L_{i,m,n}\}_{i,m,n}$ define a real quantum channel $\Lambda$. Moreover,
\begin{align}
\Lambda(\Omega_{++})=\sum_i p_i\,\Lambda_A^{(i)}(\omega_{+})\otimes\Lambda_B^{(i)}(\omega_{+})=\sum_i p_i\,\rho_A^{(i)}\otimes\rho_B^{(i)}=\rho_{\mathrm{sep}}.
\end{align}
Therefore, all separable two-qubit states are reachable from $\Omega_{++}$ via real operations.
\end{proof}
\begin{corollary}\label{cor:separable-monotone-upper-bound}
For any separable state $\rho_{\mathrm{sep}}$, we have $\mathscr F_R(\rho_{\mathrm{sep}})\leq\mathscr F_R(\Omega_{++})=1$, and $\mathscr F_r(\rho_{\mathrm{sep}})\leq\mathscr F_r(\Omega_{++})=1$. Consequently, $\Omega_{++}$ is a maximizer of both $\mathscr F_R$ and $\mathscr F_r$ over the set of separable two-qubit states.
\end{corollary}

\begin{proof}
By Theorem~\ref{thm:operational-maximal-separable}, for every $\rho_{\mathrm{sep}}\in\mathrm{Sep}$ there exists a real channel $\Lambda$ such that $\rho_{\mathrm{sep}}=\Lambda(\Omega_{++})$. Monotonicity under real operations gives 
\begin{align}
\mathscr F_K(\rho_{\mathrm{sep}})=\mathscr F_K\!\left(\Lambda(\Omega_{++})\right)\leq\mathscr F_K(\Omega_{++}),\qquad K\in\{R,r\}.
\end{align}	
A direct calculation yields $\mathscr F_R(\Omega_{++})=\mathscr F_r(\Omega_{++})=1$.
\end{proof}

Although the operationally maximal separable state is generally not unique, we choose $\Omega_{++}=|+\rangle\langle+|\otimes|+\rangle\langle+|$ as a canonical separable reference state. This choice is not arbitrary. The state $\Omega_{++}$ is operationally maximal under real operations, is a pure product state containing no entanglement, and can be prepared locally from two maximally imaginary single-qubit states. Moreover, $\mathscr F_R(\Omega_{++})=\mathscr F_r(\Omega_{++})=1$. It therefore provides a normalized, locally preparable, and reproducible resource probe for comparing the imaginarity-degrading effects of different two-qubit channels under a common input condition.

\begin{remark}
The term ``normalized'' refers to the normalization of its resource values rather than merely to the unit trace of the density operator. With logarithms taken to base $2$, the chosen reference state satisfies $\mathscr{F}_R(\Omega_{++})=\mathscr{F}_r(\Omega_{++})=1$. Consequently, $\mathscr{F}_K\big(\mathcal{E}(\Omega_{++})\big)$ directly represents the retained fraction of the reference‑state imaginarity, while $\mathcal{D}_K^{(++)}(\mathcal{E})=1-\mathscr{F}_K\big(\mathcal{E}(\Omega_{++})\big)$ represents the corresponding loss fraction.
\end{remark}

\subsection{Channel Powers and Reference-State Imaginarity Losses}
Zhang et al.~\cite{Zhang2023} introduced channel-level quantities for characterizing the ability of single-qubit quantum channels to generate and destroy imaginarity. Building on this perspective, we consider two-qubit channels acting on separable states. Because an operationally maximal separable input need not be unique, we distinguish the operational de-imaginary power, which involves an optimization over all such inputs, from the imaginarity loss evaluated at the fixed reference state $\Omega_{++}$.
\begin{definition}[Imaginary power on separable states]Let $\varepsilon$ be a two-qubit quantum channel. For $K\in\{R,r\}$, its imaginary power with respect to separable states is defined as
\begin{align}
\mathcal L_K^{\mathrm{sep}}(\varepsilon)=\max_{\rho\in\mathrm{Sep}\cap\mathcal R}\left[\mathscr F_K(\varepsilon(\rho))-\mathscr F_K(\rho)\right],
\end{align}	
where $\mathcal R$ denotes the set of real two-qubit states. Since $\mathscr F_K(\rho)=0$ for every $\rho\in\mathcal R$, this reduces to
	$\mathcal L_K^{\mathrm{sep}}(\varepsilon)=\max_{\rho\in\mathrm{Sep}\cap\mathcal R}\mathscr F_K(\varepsilon(\rho))$.
\end{definition}
Let $\mathcal M_{\mathrm{op}}^{\mathrm{sep}}=\left\{\Omega\in\mathrm{Sep}:\Omega\succeq_{\mathbb R}\rho\text{ for every }\rho\in\mathrm{Sep}\right\}$ denote the set of operationally maximal separable states.

\begin{definition}[Operational de-imaginary power]\label{def:operational-deimaginary-power}
	For $K\in\{R,r\}$, the operational de-imaginary power of a channel
	$\varepsilon$ is defined as
\begin{align}
	\mathcal D_K^{\mathrm{op}}(\varepsilon)
	=
	\sup_{
		\Omega\in
		\mathcal M_{\mathrm{op}}^{\mathrm{sep}}
	}
	\left[
	\mathscr F_K(\Omega)
	-
	\mathscr F_K(\varepsilon(\Omega))
	\right].
\end{align}	
\end{definition}

\begin{definition}[Reference-state de-imaginarity loss]
For the operationally maximal reference state $\Omega_{++}=|+\rangle\langle+|\otimes|+\rangle\langle+|$, we define $\mathcal D_K^{(++)}(\varepsilon)=\mathscr F_K(\Omega_{++})-\mathscr F_K\left(\varepsilon(\Omega_{++})\right)$, where $K\in\{R,r\}$. Since
$\mathscr F_R(\Omega_{++})=\mathscr F_r(\Omega_{++})=1$, we have $\mathcal D_K^{(++)}(\varepsilon)=1-\mathscr F_K\left(\varepsilon(\Omega_{++})\right)$.
\end{definition}
For $K\in\{R,r\}$, the output quantity $\mathscr{F}_K\big(\mathcal{E}(\Omega_{++})\big)$ represents the normalized amount of reference‑state imaginarity retained by the channel, whereas $\mathcal{D}_K^{(++)}(\mathcal{E})= 1-\mathscr{F}_K\big(\mathcal{E}(\Omega_{++})\big)$ quantifies the corresponding loss. Since $\Omega_{++}\in\mathcal{M}_{\mathrm{op}}^{\mathrm{sep}}$, the reference‑state loss satisfies $0\le \mathcal{D}_K^{(++)}(\mathcal{E})\le \mathcal{D}_K^{\mathrm{op}}(\mathcal{E})\le 1$. Hence, the quantities evaluated below are exact reference‑state losses and analytically computable lower bounds on the corresponding operational de‑imaginary powers. Equality with the operational de‑imaginary power requires either an explicit optimization over all operationally maximal separable states or an appropriate covariance property of the channel.

For completeness, we recall two information-theoretic quantities used below. For a finite ensemble $\mathcal E=\left\{p_x,\rho_x\right\}_x$, its Holevo quantity is defined as
\begin{equation}\label{eq:Holevo-quantity-definition}
\chi(\mathcal E)=S\!\left(\sum_x p_x\rho_x\right)-\sum_x p_xS(\rho_x).
\end{equation}
The Holevo bound states that the accessible information of the ensemble does not exceed $\chi(\mathcal E)$ \cite{Holevo1973}. For an equal-prior binary ensemble $\{(\frac{1}{2},\rho_0),(\frac{1}{2},\rho_1)\}$, the Shannon distinguishability is
defined as
\begin{equation}
	\label{eq:Shannon-distinguishability-definition}
	\operatorname{SD}(\rho_0,\rho_1)
	=
	\sup_{\{M_y\}}
	I(X:Y),
\end{equation}
where $X\in\{0,1\}$ labels the prepared state, $Y$ denotes the measurement outcome, and the supremum is taken over all POVMs $\{M_y\}$. Thus, $\operatorname{SD}(\rho_0,\rho_1)$ is the
accessible information of the equal-prior binary ensemble
\cite{FuchsVanDeGraaf1999}.

\begin{lemma}[Relation between $\mathscr F_r$ and $\mathscr F_R$]\label{lem:Fr-FR-relation}
For every density operator $\rho$,
\begin{equation}
		\label{eq:Fr-FR-relation}
		\mathscr F_r(\rho)
		\geq
		1-
		\mathcal H\!\left(
		\mathscr F_R(\rho)
		\right),
\end{equation}
where $\mathcal H(t)=h_2\!\left(\frac{1+t}{2}\right)$, the definition of $h_2$ is given in Eq.~\eqref{h_2(q)}.
\end{lemma}
\begin{proof}
Consider the equal-prior ensemble $\mathcal E_\rho=\left\{\left(\frac12,\rho\right),\left(\frac12,\rho^{\mathrm T}\right)\right\}$. Its Holevo quantity is
\begin{align}
\chi(\mathcal E_\rho)=S\!\left(\frac{\rho+\rho^{\mathrm T}}{2}\right)-\frac12S(\rho)-\frac12S(\rho^{\mathrm T})=S(\operatorname{Re}\rho)-S(\rho)=\mathscr F_r(\rho),\label{eq:Holevo-equals-Fr}
\end{align}
where $S(\rho^{\mathrm T})=S(\rho)$ has been used. By the Holevo bound, $\operatorname{SD}(\rho,\rho^{\mathrm T})\leq\chi(\mathcal E_\rho)$. On the other hand, Theorem~1 of Ref.~\cite{FuchsVanDeGraaf1999} gives $\operatorname{SD}(\rho,\rho^{\mathrm T})\geq1-h_2\!\left(\frac12-\frac14\left\|\rho-\rho^{\mathrm T}\right\|_1\right)$. Since $\mathscr F_R(\rho)=\frac12\left\|\rho-\rho^{\mathrm T}\right\|_1$, we obtain $\mathscr F_r(\rho)\geq1-h_2\!\left(\frac{1-\mathscr F_R(\rho)}{2}\right)$.
	
Finally, using the symmetry $h_2(q)=h_2(1-q)$, we have $h_2\!\left(\frac{
1-\mathscr F_R(\rho)}{2}\right)=h_2\!\left(\frac{1+\mathscr F_R(\rho)}{2}\right)=\mathcal H\!\left(\mathscr F_R(\rho)\right)$, which proves Eq.~\eqref{eq:Fr-FR-relation}.
\end{proof}

\subsection{Reference-State Losses and Exact Operational Results for Selected Two-Qubit Channels}

In this subsection, reference-state losses are evaluated for all channels considered below. For the $PD$, $PF$, and $BF$ channels, we additionally perform the optimization over $\mathcal M_{\mathrm{op}}^{\mathrm{sep}}$ and obtain the exact operational de-imaginary powers. For brevity, the superscript $(++)$ is omitted only for the reference-state quantities in the corresponding formulas and figures. The optimized quantities are always written explicitly as $\mathcal D_K^{\mathrm{op}}$.

\subsubsection{Phase-Damping, Phase-Flip, and Bit-Flip Channels}
We first consider the two-qubit phase damping ($PD$), two-qubit phase-flipping ($PF$) and two-qubit bit-flipping ($BF$) channels. To keep the main text concise, we provide the explicit matrix representations of their Kraus operators in Appendix~\ref{app:kraus-operators}; see Eqs.~\eqref{eq:PD}, \eqref{eq:PF} and \eqref{eq:BF}.

The Kraus operators of $PD$, $PF$ and $BF$ are obtained by taking the tensor product of the Kraus operators of the corresponding single-qubit channels respectively. Through calculation, it can be found that the above three groups of operators all satisfy the completeness condition: $\sum_{i,j\in\left \{ 0,1 \right \} }^{}K_{ij}^{\dagger }K_{ij}=I_4$.

Since the Kraus operators of the two-qubit phase damping ($PD$), two-qubit phase-flipping ($PF$) and two-qubit bit-flipping ($BF$) channels are all real, a real separable state remains a real state after passing through these channel. Therefore, the imaginary powers of the two-qubit $PD$, $PF$ and $BF$ channels with respect to separable states are 0.

For the common operationally maximal reference state $\Omega_{++}=|+\rangle\langle+|\otimes|+\rangle\langle+|$, the evolved density matrices under the PD, PF, and BF channels are given in Appendix~\ref{app:evolved-density-matrices}, Eqs. \eqref {eq:pd-state-evol}, \eqref {eq:pf-state-evol} and \eqref {eq:bf-state-evol}.

Evaluating the robustness‑ and relative‑entropy‑based imaginarity losses of the common reference state $\Omega_{++}$, we obtain the following expressions.
\begin{equation}
		\begin{aligned}
			\mathcal{D}_{R }(\varepsilon _{PD})=&1-\frac{\alpha _{1}+ \alpha _{2}+\left | \alpha _{1}-\alpha _{2} \right | }{2},\\
			\mathcal{D}_{R }(\varepsilon _{PF})=&\mathcal{D}_{R }(\varepsilon _{BF})=1-\left | p_{1} +p_{2}-1 \right |-\left | p_{1}-p _{2}  \right | ;			
		\end{aligned}
\end{equation}
\begin{equation}
		\begin{aligned}
			\mathcal{D}_{r }(\varepsilon _{PD})=&\mathcal H(\alpha_1)+\mathcal H(\alpha_2)
			-\mathcal H(\alpha_1\alpha_2),\\
			\mathcal{D}_{r }(\varepsilon _{PF})=\mathcal{D}&_{r }(\varepsilon _{BF})=\mathcal H(|t_1|)+\mathcal H(|t_2|)-\mathcal H(|t_1t_2|),			
		\end{aligned}
\end{equation}
where $\alpha _{i}=\sqrt{1-\gamma _{i} }$, $t_i=2p_i-1$, $i=1,2$, and the function $\mathcal{H}(\cdot)$ is given by Eq.~\eqref{H_t}.
\begin{figure}[h]
	\centering
	\includegraphics[width=0.9\linewidth]{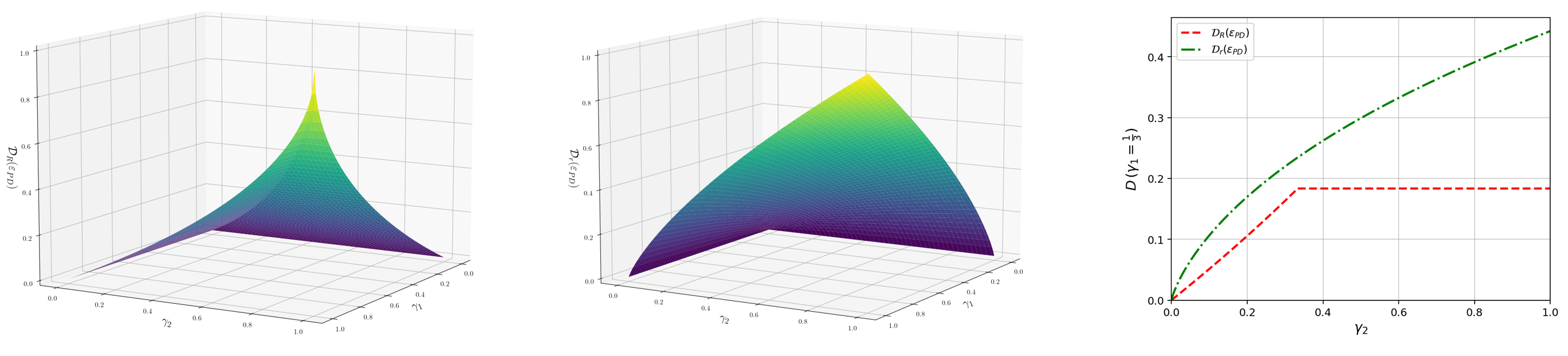}
	\caption{Reference-state imaginarity losses under a two-qubit phase-damping channel.}
	\label{fig:6pd}
\end{figure}
\begin{figure}[h]
	\centering
	\includegraphics[width=0.9\linewidth]{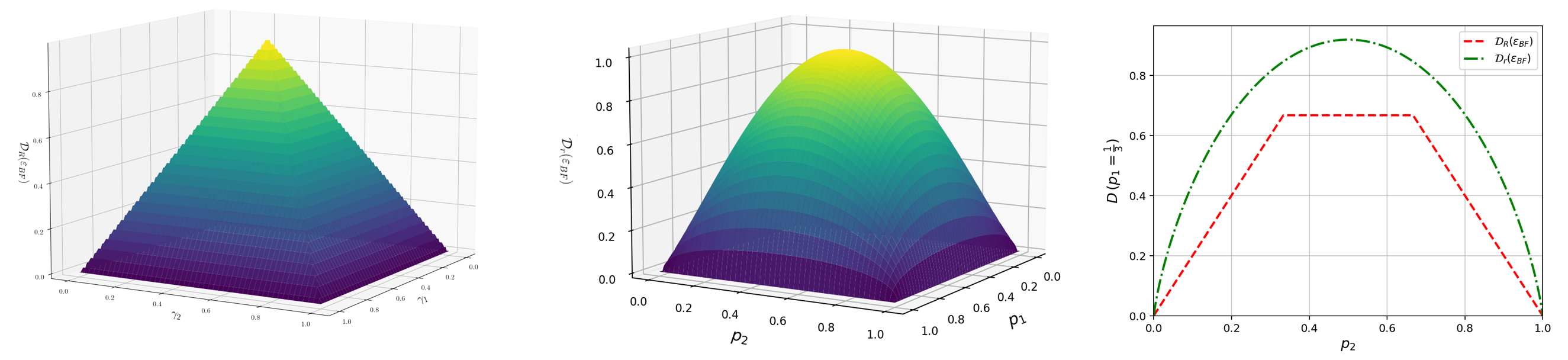}
	\caption{Reference-state imaginarity losses under a two-qubit bit-flip channel.}
	\label{fig:7pf-and-bf}
\end{figure}

Figure~\ref{fig:6pd} shows the robustness- and relative-entropy-based reference-state losses of the two-qubit phase-damping channel. Both quantities vanish whenever either $\gamma_1=0$ or $\gamma_2=0$, and both attain unity at $\gamma_1=\gamma_2=1$. For the cross-section $\gamma_1=1/3$, $\mathcal D_R(\varepsilon_{\mathrm{PD}})$ increases with $\gamma_2$ on $0\leq\gamma_2\leq\frac13$ and then remains constant for $\frac13\leq\gamma_2\leq1$. Its plateau value is $1-\sqrt{\frac23}$. In contrast,
$\mathcal D_r(\varepsilon_{\mathrm{PD}})$ continues to increase monotonically throughout the whole interval $0\leq\gamma_2\leq1$. Thus, the robustness-based loss saturates once $\gamma_2$ exceeds $\gamma_1$, whereas the relative-entropy-based loss remains sensitive to further phase damping.

Figure~\ref{fig:7pf-and-bf} shows the reference-state losses of the two-qubit bit-flip channel. Both surfaces are symmetric under the interchange $p_1\leftrightarrow p_2$. The two losses vanish on the boundary where at least one of the parameters is equal to $0$ or $1$, and they attain unity at $p_1=p_2=\frac12$. For the cross-section $p_1=1/3$, the robustness-based loss is
\begin{equation}
\begin{aligned}
\mathcal D_R(\varepsilon_{\mathrm{BF}})=
\begin{cases}
	2p_2,
	&
	0\leq p_2\leq\frac13,
	\\[1mm]
	\displaystyle\frac23,
	&
	\frac13\leq p_2\leq\frac23,
	\\[1mm]
	2(1-p_2),
	&
	\frac23\leq p_2\leq1.
\end{cases}
\end{aligned}
\end{equation}
It therefore increases, remains constant over the interval $[1/3,2/3]$, and then decreases. The relative-entropy-based loss is symmetric about $p_2=1/2$; it increases for $0\leq p_2\leq1/2$, decreases for $1/2\leq p_2\leq1$, and reaches its maximum at $p_2=1/2$.

The preceding results characterize the reference‑state imaginarity losses evaluated on the physically feasible and analytically tractable state $\Omega_{++}$, which provides rigorous lower bounds and intuitive benchmark comparisons for the imaginarity‑degrading capabilities of quantum channels. We now carry out the exact evaluation over the entire set $\mathcal{M}_{\mathrm{op}}^{\mathrm{sep}}$ from Definition~\ref{def:operational-deimaginary-power}, and obtain the exact operational de‑imaginary powers of the $PD$, $PF$, and $BF$ channels.
\begin{proposition}[Operational de-imaginary powers of the PD, PF,
	and BF channels]
\label{prop:operational-deimaginary-PD-PF-BF}
Let $\alpha_i=\sqrt{1-\gamma_i}$, $t_i=2p_i-1$, where $i=1,2$, and define
\begin{equation}
\begin{aligned}
	c_{\mathrm{PD}}=&\alpha_1\alpha_2=\sqrt{(1-\gamma_1)(1-\gamma_2)},\\
	c_{\mathrm{F}}=&|t_1t_2|=\left|(2p_1-1)(2p_2-1)\right|.
\end{aligned}
\end{equation}
Then the operational de-imaginary powers in Definition~\ref{def:operational-deimaginary-power} are
\begin{equation}
\begin{aligned}
\mathcal D_R^{\mathrm{op}}\left(\varepsilon_{\mathrm{PD}}\right)&=1-c_{\mathrm{PD}},		\\
\mathcal D_r^{\mathrm{op}}\left(\varepsilon_{\mathrm{PD}}\right)&=\mathcal H(c_{\mathrm{PD}}),\\
		\mathcal D_R^{\mathrm{op}}
		\left(
		\varepsilon_{\mathrm{PF}}
		\right)
		&=
		\mathcal D_R^{\mathrm{op}}
		\left(
		\varepsilon_{\mathrm{BF}}
		\right)
		=
		1-c_{\mathrm{F}},
		\\
		\mathcal D_r^{\mathrm{op}}
		\left(
		\varepsilon_{\mathrm{PF}}
		\right)
		&=
		\mathcal D_r^{\mathrm{op}}
		\left(
		\varepsilon_{\mathrm{BF}}
		\right)
		=
		\mathcal H(c_{\mathrm{F}}),
\end{aligned}
\end{equation}
where $\mathcal H(c)=h_2\left(\frac{1+c}{2}\right)$, the definition of $h_2$ is given in Eq.~\eqref{h_2(q)}. The maxima are attained by the operationally maximal separable states $\Omega_{\mathrm{PD}}^\star=\Omega_{\mathrm{PF}}^\star=\frac14\left(I\otimes I+Y\otimes X\right)$ and $\Omega_{\mathrm{BF}}^\star=\frac14\left(I\otimes I+Y\otimes Z\right)$, respectively.
\end{proposition}
The proof of Proposition~\ref{prop:operational-deimaginary-PD-PF-BF} is given in Appendix~\ref{app:proof-operational-deimaginary}.

\subsubsection{Amplitude-Damping and Phase-Amplitude-Damping Channels}
Next, we consider the two-qubit amplitude damping ($AD$) and phase-amplitude damping ($PAD$) channels. The detailed matrix forms of their Kraus operators are summarized in Appendix~\ref{app:kraus-operators}; see Eqs.~\eqref{eq:AD} and \eqref{eq:PAD}.

Since the Kraus operators of the two-qubit amplitude damping ($AD$) and phase-amplitude damping ($PAD$) channels are all real, a real separable state remains a real state after passing through these channel. Therefore, the imaginary powers of the two-qubit $AD$ and $PAD$ channels with respect to separable states are 0.

For the common reference state $\Omega_{++}$, the evolved density matrices under the AD and PAD channels are given in the Appendix~\ref{app:evolved-density-matrices}; see Eqs.~\eqref {eq:ad-state-evol} and~\eqref {eq:pad-state-evol}.

For the robustness of imaginarity, direct evaluation of the reference-state outputs gives
\begin{equation}
		\begin{aligned}
			\mathcal{D}_{R }(\varepsilon _{AD})=1&-\frac{1}{2} \sqrt{\beta _{1}(\gamma _{2}^{2}+1 )+\beta _{2}(\gamma _{1}^{2}+1 )-2\sqrt{\beta _{1}\beta _{2}(\gamma _{1}^{2}\gamma _{2}^{2}+1)+\gamma _{1}^{2}\beta _{2} ^{2}+\gamma _{2}^{2}\beta _{1} ^{2} } }\\
			&- \frac{1}{2} \sqrt{\beta _{1}(\gamma _{2}^{2}+1 )+\beta _{2}(\gamma _{1}^{2}+1 )+2\sqrt{\beta _{1}\beta _{2}(\gamma _{1}^{2}\gamma _{2}^{2}+1)+\gamma _{1}^{2}\beta _{2} ^{2}+\gamma _{2}^{2}\beta _{1} ^{2} } } ,\\
			\mathcal{D}_{R }(\varepsilon _{PAD})=1&-\frac{\sqrt{-\gamma_{1}\gamma_{2}^{2}-\gamma_{1}+\gamma_{2}^{2}- \gamma_{2}+2\sqrt{\gamma_{2}^{2}\beta _{1}^{2}+\beta _{1}\beta _{2} } +2 } }{2}\\
			&-\frac{\sqrt{-\gamma_{1}\gamma_{2}^{2}-\gamma_{1}+\gamma_{2}^{2}- \gamma_{2}-2\sqrt{\gamma_{2}^{2}\beta _{1}^{2}+\beta _{1}\beta _{2} } +2 } }{2},			
		\end{aligned}
\end{equation}
where $\beta _{i}=1-\gamma _{i}$, $i=1,2$.

For the relative entropy of imaginarity, direct evaluation of the reference-state outputs gives
\begin{equation}
	\label{eq:Dr-AD-expanded}
	\begin{aligned}
		\mathcal{D}_{r}
		\left(
		\varepsilon_{\mathrm{AD}}
		\right)
		={}&
		1
		-\frac{1+r_1}{2}\log{(1+r_1)}
		-\frac{1-r_1}{2}\log{(1-r_1)}+\frac{\xi_{+}}{4}\log{\xi_{+}}
		+\frac{\xi_{-}}{4}\log{\xi_{-}}
		\\
		&-\frac{1+r_2}{2}\log{(1+r_2)}
		-\frac{1-r_2}{2}\log{(1-r_2)}+\frac{\zeta_{+}}{4}\log{\zeta_{+}}
		+\frac{\zeta_{-}}{4}\log{\zeta_{-}},
	\end{aligned}
\end{equation}
where \begin{equation}
	\begin{aligned}
		r_i=&\sqrt{1-\gamma_i+\gamma_i^2},\qquad i=1,2,\\
		\xi_{\pm}=&1+\gamma_1\gamma_2\pm \sqrt{(\gamma_1+\gamma_2)^2+(1-\gamma _{1})(1-\gamma _{2})},\\
		\zeta_{\pm}=&1-\gamma_1\gamma_2\pm \sqrt{(\gamma_1-\gamma_2)^2+(1-\gamma _{1})(1-\gamma _{2})};
	\end{aligned}
\end{equation}
\begin{equation}
	\label{eq:Dr-PAD-expanded}
	\begin{aligned}
		\mathcal{D}_{r}\left(\varepsilon_{\mathrm{PAD}}\right)=\mathcal H(u)+\mathcal H(v)-\mathcal H(w),
	\end{aligned}
\end{equation}
where $u=\sqrt{1-\gamma_1}$, $v=\sqrt{1-\gamma_2+\gamma_2^2}$, $w=\sqrt{\gamma_2^2+(1-\gamma_1)(1-\gamma_2)}$.
\begin{figure}[H]
	\centering
	\includegraphics[width=0.9\linewidth]{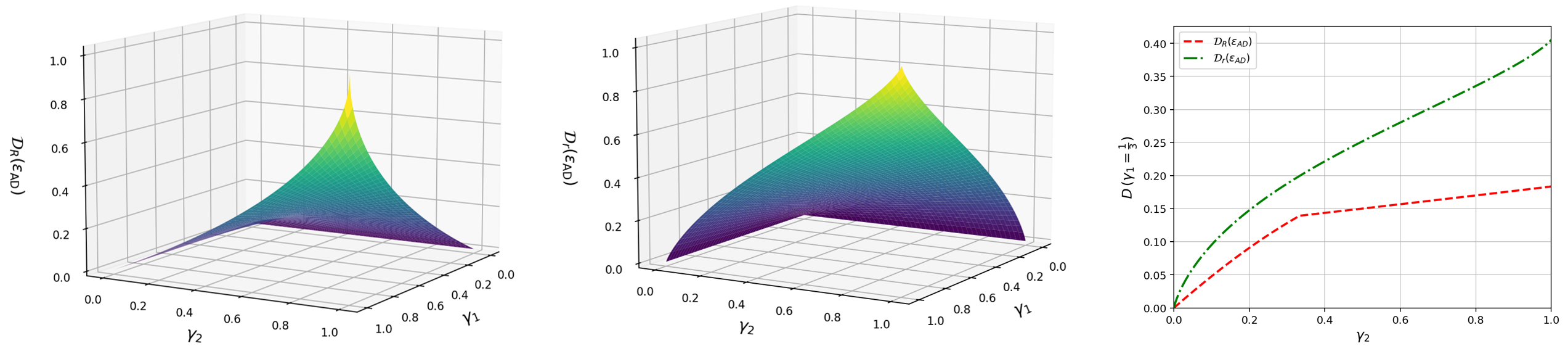}
	\caption{Reference-state imaginarity losses under a two-qubit amplitude-damping channel.}
	\label{fig:8ad}
\end{figure}
\begin{figure}[H]
	\centering
	\includegraphics[width=0.9\linewidth]{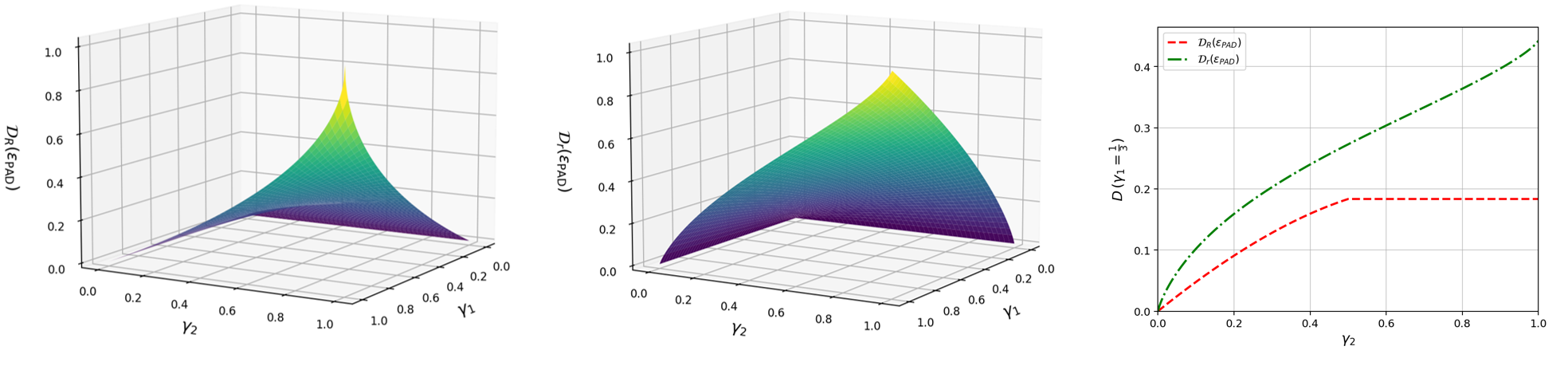}
	\caption{Reference-state imaginarity losses under a two-qubit phase-amplitude-damping channel.}
	\label{fig:9pad}
\end{figure}

Figure~\ref{fig:8ad} shows the reference-state losses of the two-qubit amplitude-damping channel. Both surfaces are symmetric under the interchange $\gamma_1\leftrightarrow\gamma_2$. The two losses vanish whenever either $\gamma_1=0$ or $\gamma_2=0$, and both attain unity at $\gamma_1=\gamma_2=1$. For the cross-section $\gamma_1=1/3$, both $\mathcal D_R(\varepsilon_{\mathrm{AD}})$ and $\mathcal D_r(\varepsilon_{\mathrm{AD}})$ increase monotonically with $\gamma_2$. The robustness-based curve is concave on $0\leq\gamma_2\leq\frac13$, has a kink at $\gamma_2=\gamma_1=\frac13$, and becomes slightly convex on $\frac13\leq\gamma_2\leq1$. The relative-entropy-based loss varies smoothly and lies above the robustness-based loss throughout the displayed nonzero-noise cross-section.

Figure~\ref{fig:9pad} shows the reference-state losses of the two-qubit phase-amplitude-damping channel defined in Eq.~\eqref{eq:PAD}. Since $\gamma_1$ and $\gamma_2$ characterize different
local damping mechanisms, the two surfaces are not symmetric under the interchange $\gamma_1\leftrightarrow\gamma_2$. Both losses vanish whenever either $\gamma_1=0$ or $\gamma_2=0$,
and both attain unity at $\gamma_1=\gamma_2=1$. For the cross-section $\gamma_1=1/3$, $\mathcal D_R(\varepsilon_{\mathrm{PAD}})$ increases with $\gamma_2$ on $0\leq\gamma_2\leq\frac12$ and then remains constant for $\frac12\leq\gamma_2\leq1$. Its plateau value is $1-\sqrt{\frac23}$. In contrast,
$\mathcal D_r(\varepsilon_{\mathrm{PAD}})$ increases monotonically throughout the entire interval $0\leq\gamma_2\leq1$. For the displayed cross-section, the relative-entropy-based loss is
larger than the robustness-based loss whenever $\gamma_2>0$.

\subsubsection{Bit-Phase-Flip and Depolarizing Channels}
Finally, we discuss two channels that are slightly different from the above channels: the two-qubit bit-phase flipping ($BPF$) channel and the two-qubit depolarizing ($DEP$) channel. 

The bit-phase flip channel \cite{Nielsen10} is a quantum channel that includes both bit flip errors and phase flip errors. For a two-qubit system, we assume that each qubit independently undergoes a bit-phase flip. The detailed matrix forms of its Kraus operators are summarized in Appendix~\ref{app:kraus-operators}; see Eq.~\eqref{eq:BPF}.

Let $\rho$ be an arbitrary real state. Through calculation, we can find that \begin{align} \varepsilon _{BPF}(\rho)=\sum_{i,j\in\left \{ 0,1 \right \} }^{}K_{ij}^{BPF}(\rho){K_{ij}^{BPF}}^{\dagger } \end{align} is a real matrix. Therefore  the imaginary powers of the $BPF$ channel with respect to separable states based on the robustness of imaginarity, and relative entropy of imaginarity are both zero.

The state $\Omega_{++}=\ket{+}\bra{+} \otimes \ket{+}\bra{+}$ after passing through the $BPF$ channel becomes:
\begin{align}
\varepsilon _{BPF}(\ket{+}\bra{+}\otimes\ket{+}\bra{+})=\sum_{i,j\in\left \{ 0,1 \right \} }^{}K_{ij}^{BPF}(\ket{+}\bra{+}\otimes\ket{+}\bra{+}){K_{ij}^{BPF}}^{\dagger }=\frac{1}{4}\begin{bmatrix}\begin{smallmatrix}
		1&-i  &-i  &-1 \\
		i&1  &1  &-i \\
		i&1  &1  &-i \\
		-1&i  &i  &1
\end{smallmatrix}	\end{bmatrix}.
\end{align}
It can be seen that the reference state $\Omega_{++}$ remains unchanged under the BPF channel. Therefore, the reference‑state losses vanish: $\mathcal{D}_R\bigl(\mathcal{E}_{\mathrm{BPF}}\bigr)=\mathcal{D}_r\bigl(\mathcal{E}_{\mathrm{BPF}}\bigr)=0$. This result concerns the selected reference state and does not, by
itself, imply that the global operational de‑imaginary powers vanish.

The depolarizing channel \cite{Nielsen10} is an important type of quantum noise. A qubit is replaced by a completely mixed state $\frac{I}{2}$ with probability $p$, and remains unchanged with probability $1-p$. For the two-qubit case, the two-qubit depolarizing channel will make the first qubit become the completely mixed state $\frac{I}{2}$ with probability $p_1$, and make the second qubit become the completely mixed state $\frac{I}{2}$ with probability $p_2$. Obviously, a real separable state remains a real state after passing through the two-qubit depolarizing channel. Therefore, the two-qubit depolarizing channel has zero imaginary power with respect to separable states for each of the two measures.

The state $\Omega_{++}=\ket{+}\bra{+} \otimes \ket{+}\bra{+}$ after passing through the $DEP$ channel becomes:
\begin{align}
\varepsilon _{DEP}(\ket{+}\bra{+}\otimes\ket{+}\bra{+})=\frac{1}{4} \begin{bmatrix}\begin{smallmatrix}
		1&-i(1-p_2 )&-i(1-p_1 )&-(1-p_1 )(1-p_2 )\\
		i(1-p_2)&1&(1-p_1)(1-p_2)&-i(1-p_1)\\
		i(1-p_1)&(1-p_1)(1-p_2)&1&-i(1-p_2)\\
		-(1-p_1)(1-p_2)&i(1-p_1)&i(1-p_2)&1
\end{smallmatrix}\end{bmatrix}
\end{align}
Evaluating the reference-state losses for the output state, we obtain
\begingroup
\small
\begin{equation}
\begin{aligned}
\mathcal{D}_{R }(\varepsilon _{DEP})=&\frac{p_{1}+p_{2}-\left | p_{1}-p_{2} \right |  }{2} ,\\
\mathcal{D}_{r }(\varepsilon _{DEP})=&\frac{1-p_{1}^{\prime }p_{2}^{\prime } }{2} \log_{}{[1-p_{1}^{\prime }p_{2}^{\prime } ] }+\frac{1+p_{1}^{\prime }p_{2}^{\prime } }{2} \log_{}{[1+p_{1}^{\prime }p_{2}^{\prime } ] }-\frac{(2-p_{1} )(2-p_{2} ) }{4}\log_{}{[(2-p_{1} )(2-p_{2} )]}\\
&-\frac{p_{1}(2-p_{2} ) }{4}\log_{}{[p_{1}(2-p_{2} )]}-\frac{p _{2}(2-p _{1})}{4} \log_{}{[p _{2}(2-p _{1})]}-\frac{p_{1} p_{2} }{4}\log_{}{(p_{1} p_{2})}+1,	
\end{aligned}
\end{equation}
\endgroup
where $p_{1}^{\prime }=1-p_{1}$, $p_{2}^{\prime }=1-p_{2}$.
\begin{figure}[H]
	\centering
	\includegraphics[width=0.9\linewidth]{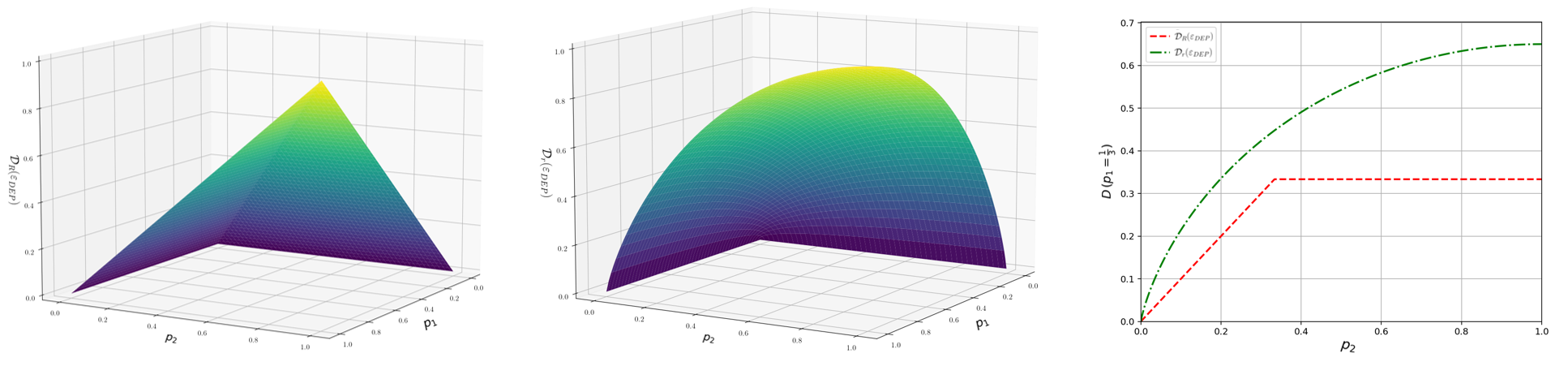}
	\caption{Reference-state imaginarity losses under a two-qubit depolarization channel.}
	\label{fig:f10dep}
\end{figure}

Figure~\ref{fig:f10dep} shows the reference-state losses of the two-qubit depolarizing channel. Both surfaces are symmetric under the interchange $p_1\leftrightarrow p_2$. The two losses vanish whenever either $p_1=0$ or $p_2=0$, and both attain unity at $p_1=p_2=1$. For the cross-section $p_1=1/3$, the robustness-based loss reduces to $\mathcal D_R(\varepsilon_{\mathrm{DEP}})=\min\left\{\frac13,p_2\right\}$. It therefore increases linearly on $0\leq p_2\leq\frac13$ and remains constant for $\frac13\leq p_2\leq1$. The relative-entropy-based loss increases monotonically with $p_2$ and gradually approaches $\mathcal H\left(\frac23\right)$ at $p_2=1$. Thus, unlike the robustness-based loss, it remains sensitive to the increase of $p_2$ after the robustness plateau has been reached.

The comparisons in Figs.~\ref{fig:6pd}--\ref{fig:f10dep} further show that the two normalized monotones need not respond identically to the same channel parameters: the robustness-based loss may develop plateaus or kinks, whereas the relative-entropy-based loss can remain smooth and parameter-sensitive. Their relative decay strength is therefore channel- and parameter-dependent.

\section{Conclusion}

In this work, we have systematically investigated the decay behavior of three imaginarity measures: the $l_1$-norm-based imaginarity measure, the robustness of imaginarity, and the relative entropy of imaginarity—under several representative quantum noise channels. For single-qubit systems, we studied arbitrary pure states under the dephasing, generalized amplitude-damping, and phase-amplitude-damping channels. We established an exact attenuation theorem showing that, under any real single-qubit quantum channel, the $l_1$-norm-based imaginarity and the robustness of imaginarity are equal and are reduced by the same channel-dependent factor. The relative entropy of imaginarity exhibits a different behavior: although its initial value is fixed by the parameter $A=|\langle\psi^*|\psi\rangle|$, its decay under a basis-dependent channel generally also depends on the orientation of the input state relative to the fixed reference basis. These results clarify both the agreement and the complementary sensitivities of different imaginarity measures.
	
For two-qubit systems, we further analyzed the decay of imaginarity for representative entangled states and dual-rail states. To provide a concrete physical interpretation, we considered a photonic dual-rail qubit subject to photon-loss noise modeled by symmetric amplitude damping. In this setting, photon loss transfers the state from the logical dual-rail subspace to the vacuum state, which represents a detectable erasure event. The three imaginarity measures considered in this work predict the same fractional loss of the initial resource, determined by the photon-loss probability. Moreover, conditioned on the no-loss outcome, the normalized logical state and its imaginarity remain unchanged. This example illustrates how the analytical results can be used to quantify the robustness of imaginarity in a specific physical setting.
	
For separable two-qubit systems, we introduced a real-operation preorder and proved that the product state $\Omega_{++}=|+\rangle\langle+|\otimes|+\rangle\langle+|$ is operationally maximal. This state provides a normalized, locally preparable, and reproducible reference for comparing different two-qubit channels. We then distinguished the imaginary power on separable states, the operational de-imaginary power, and the reference-state imaginarity loss. The imaginary powers of the real channels considered in this work are zero. For the two-qubit phase-damping, phase-flip, and bit-flip channels, we obtained exact operational de-imaginary powers based on both the robustness and relative entropy of imaginarity. For the amplitude-damping, phase-amplitude-damping, bit-phase-flip, and depolarizing channels, we obtained analytical reference-state losses. These quantities provide common benchmark comparisons and rigorous lower bounds on the corresponding operational de-imaginary powers, rather than full global optimizations.
	
Several questions remain open. Our two-qubit analysis has been carried out for selected pure states and specific noise models, and it is not yet clear to what extent the observed decay behavior persists for general mixed states or higher-dimensional systems. The same framework may also be applied to asymmetric, correlated, and non-Markovian channels. At the channel level, the exact operational de-imaginary powers of the amplitude-damping, phase-amplitude-damping, bit-phase-flip, and depolarizing channels have not yet been obtained. A related question is whether, and under what conditions, the reference-state losses derived above coincide with the corresponding operational optima. It would also be worthwhile to examine whether the real-operation preorder introduced for separable two-qubit states admits a useful extension to multipartite systems and entangled inputs.

\section*{\centering Data availability statement}
No new data were created or analysed in this study.

\section*{\centering Conflict of interest}
The authors declare that they have no conflicts of interest.

\section*{\centering Ethics statement}
This article does not contain any studies involving human participants or animals performed by any of the authors. No ethical approval or informed consent was required for this study, as it is purely theoretical in nature and does not involve any experimental data collection, human subjects, animal research, or sensitive data handling.

\clearpage
\fancyhead{}
\fancyhead[C]{REFERENCES}

\clearpage
\fancyhead{}
\fancyhead[C]{APPENDIX}
\appendix
\section*{Appendix}
The appendices collect technical details that support the results presented in the main text. To improve the readability of the manuscript, lengthy Kraus-operator representations, explicit evolved density matrices, and detailed mathematical proofs are placed here, while the main text focuses on the principal analytical results, their comparison, and their physical interpretation.

\section{Proofs of Theorem 3.1, Corollary 3.2, and Proposition 4.3}\label{app:proofs}
\renewcommand{\theequation}{A\arabic{equation}}
\renewcommand*{\theHequation}{A.\arabic{equation}}
\setcounter{equation}{0}
For completeness, we collect in this appendix the detailed proofs of Theorem~\ref{thm:real-qubit-attenuation}, Corollary~\ref{cor:canonical-state-reduction}, and Proposition~\ref{prop:operational-deimaginary-PD-PF-BF}. These derivations are separated from the main text in order to keep the presentation focused on the main physical results and the comparison among different imaginarity measures. These proofs further clarify why the $\ell_1$-norm‑based imaginarity and the robustness of imaginarity exhibit identical attenuation behavior for single-qubit states under arbitrary real quantum channels, and support the operational characterization of de-imaginary power established in Proposition~4.3.

\subsection{Proof of Theorem 3.1}\label{app:proof-real-qubit-attenuation}
An arbitrary single-qubit density operator can be written in the Bloch form as
\begin{align}
	\rho=\frac{1}{2}\left(I+x\sigma_x+y\sigma_y+z\sigma_z\right)=
	\frac{1}{2}\begin{pmatrix}
		1+z & x-iy\\
		x+iy & 1-z
	\end{pmatrix},
\end{align}
where $x,y,z\in\mathbb{R}$ and $x^2+y^2+z^2\leq 1$. The real vector $\boldsymbol{r}=(x,y,z)$ is called the Bloch vector of the state. The set of all single-qubit
density operators is represented by the Bloch ball: mixed states lie
inside the ball, whereas pure states satisfy $x^2+y^2+z^2=1$ and lie on the surface of the Bloch sphere. In the fixed computational
basis, the $y$-component is the only component responsible for the
imaginary entries of the density matrix.

Let $J=\begin{pmatrix}
	0 & -1\\
	1 & 0
\end{pmatrix}$. The entrywise imaginary part of $\rho$ is $\operatorname{Im}\rho=\frac{y}{2}J$. Therefore,
\begin{align}
	\mathscr F_{l_{1}}(\rho)=|\operatorname{Im}\rho_{01}|+|\operatorname{Im}\rho_{10}|=\left|-\frac{y}{2}\right|+\left|\frac{y}{2}\right|=|y|.
\end{align}
Moreover,
\begin{align}
	\rho-\rho^T=iyJ.
\end{align}
Since the two singular values of $J$ are both equal to $1$, the two
singular values of $\rho-\rho^T$ are both equal to $|y|$. Hence,
\begin{align}
	\mathscr	F_{R}(\rho)=\frac{1}{2}\|\rho-\rho^T\|_1=|y|.
\end{align}
Consequently, $\mathscr F_{l_{1}}(\rho)=\mathscr F_R(\rho)=|y|$. For every real $2\times2$ matrix $K$, a direct calculation gives $KJK^T=(\det K)J$.
Since all Kraus operators of $\varepsilon$ are real, it follows that
\begin{equation}
	\begin{aligned}
		\operatorname{Im}\varepsilon(\rho)
		&=
		\sum_m K_m(\operatorname{Im}\rho)K_m^T=\frac{y}{2}\sum_m K_mJK_m^T\\
		&=
		\frac{y}{2}\sum_m(\det K_m)J=\frac{\eta_\varepsilon y}{2}J,
	\end{aligned}
\end{equation}
where $\eta_\varepsilon=\sum_m\det K_m$.
Thus, the $y$-component of the Bloch vector is transformed according
to $y\longmapsto \eta_\varepsilon y$. It follows that
\begin{align}
	\mathscr F_{l_{1}}\!\left(\varepsilon(\rho)\right)=|\eta_\varepsilon|\mathscr F_{l_{1}}(\rho),
\end{align}
and
\begin{align}
	\mathscr	F_{R}\!\left(\varepsilon(\rho)\right)=|\eta_\varepsilon|\mathscr F_{R}(\rho).
\end{align}

The preceding relations show that the action of the real channel on
$\mathscr F_{l_1}$ and $\mathscr F_R$ is completely characterized by the
channel-dependent factor $|\eta_{\varepsilon}|$. It remains to verify
that this factor cannot exceed unity.

Let the two columns of the real Kraus operator $K_m$ be denoted by
$\boldsymbol{a}_m$ and $\boldsymbol{b}_m$, respectively, so that $K_m=\begin{bmatrix}
	\boldsymbol{a}_m & \boldsymbol{b}_m
\end{bmatrix}$.
Since $\varepsilon$ is trace preserving, its Kraus operators satisfy $\sum_m K_m^T K_m = I_{2}$. By comparing the two diagonal entries on both sides, we obtain
\begin{align}
	\sum_m \|\boldsymbol{a}_m\|^2 = 1,
	\qquad
	\sum_m \|\boldsymbol{b}_m\|^2 = 1.
\end{align}
For every real $2\times 2$ matrix $K_m$, the Gram determinant identity
gives
\begin{align}
	|\det K_m|^2=\|\boldsymbol{a}_m\|^2\|\boldsymbol{b}_m\|^2-|\boldsymbol{a}_m^T\boldsymbol{b}_m|^2.
\end{align}
Therefore, $|\det K_m|\leq\|\boldsymbol{a}_m\|\|\boldsymbol{b}_m\|$.
Using the triangle inequality and the Cauchy--Schwarz inequality, we
find that
\begin{equation}
	\begin{aligned}
		|\eta_{\varepsilon}|
		&=
		\left|
		\sum_m \det K_m
		\right|
		\leq
		\sum_m |\det K_m|\\
		&\leq
		\sum_m
		\|\boldsymbol{a}_m\|
		\|\boldsymbol{b}_m\|
		\leq
		\sqrt{\sum_m\|\boldsymbol{a}_m\|^2}\sqrt{\sum_m\|\boldsymbol{b}_m\|^2}=1.
	\end{aligned}
\end{equation}
Hence, we obtain $0\leq |\eta_{\varepsilon}|\leq 1$.

Combining this bound with the scaling relations established above, we
obtain
\begin{equation}
	\begin{aligned}
		\Delta I_{l_1}^{\varepsilon}(\rho)
		=\mathscr F_{l_1}(\rho)-\mathscr F_{l_1}\!\left(\varepsilon(\rho)\right)
		=\left(1-|\eta_{\varepsilon}|\right)\mathscr F_{l_1}(\rho)
		=\left(1-|\eta_{\varepsilon}|\right)|y|
		\geq 0,
	\end{aligned}
\end{equation}
and similarly,
\begin{equation}
	\begin{aligned}
		\Delta I_R^{\varepsilon}(\rho)
		=\mathscr F_R(\rho)-\mathscr F_R\!\left(\varepsilon(\rho)\right)
		=\left(1-|\eta_{\varepsilon}|\right)\mathscr F_R(\rho)
		=\left(1-|\eta_{\varepsilon}|\right)|y|
		\geq 0.
	\end{aligned}
\end{equation}
Thus, every real single-qubit quantum channel is non-increasing with
respect to both $\mathscr F_{l_1}$ and $\mathscr F_R$.

Finally, we express the above general result in terms of the pure-state parameter $A$ used in Eq.~\eqref{a pure state under a real orthogonal transformation}. Let $\rho=|\psi\rangle\langle\psi|$ be a pure single-qubit state. Its Bloch vector satisfies $x^2+y^2+z^2=1$. Moreover,
\begin{align}
	A^2=|\langle\psi^*|\psi\rangle|^2=\operatorname{Tr}(\rho\rho^T)=\frac{1}{2}\left(1+x^2-y^2+z^2\right)=x^2+z^2.
\end{align}
In the last equality, we have used the pure-state condition $x^2+y^2+z^2=1$.
Consequently,
\begin{align}
	y^2=1-A^2,
\end{align}
and hence
\begin{align}
	\mathscr	F_{l_1}(\rho)=\mathscr F_R(\rho)=|y|=\sqrt{1-A^2}.
\end{align}
Therefore, for every pure single-qubit state,
\begin{align}	
	\Delta I_{l_1}^{\varepsilon}=\Delta I_R^{\varepsilon}=\left(1-|\eta_{\varepsilon}|\right)\sqrt{1-A^2}.
\end{align}
This completes the proof.

\subsection{Proof of Corollary 3.2}\label{app:proof-canonical-state-reduction}
Define the channel induced by the real orthogonal matrix $O$ as $\mathcal{O}(\cdot)=O(\cdot) O^T$. Since $O^TO=I_2$ and $O$ is real, $\mathcal{O}$ is a real single-qubit quantum channel with the single Kraus operator $O$. Its attenuation coefficient, as defined in Theorem \ref{thm:real-qubit-attenuation}, is $\eta_{\mathcal{O}}=\det O$. Because $O$ is orthogonal, we have $\det O=\pm1$ and hence $|\eta_{\mathcal{O}}|=|\det O|=1$. Applying Theorem \ref{thm:real-qubit-attenuation} to $\mathcal{O}$ gives 
\begin{equation} \begin{aligned} \mathscr F_K(O\rho O^T) = \mathscr F_K\!\left(\mathcal{O}(\rho)\right)= |\eta_{\mathcal{O}}|\mathscr F_K(\rho)= \mathscr F_K(\rho), \end{aligned} \qquad K\in\{l_1,R\}. 
\end{equation} 
Thus, both $\mathscr F_{l_1}$ and $\mathscr F_R$ are invariant under real orthogonal transformations of a single-qubit state. We now apply Theorem~\ref{thm:real-qubit-attenuation} to the real channel $\varepsilon$. Denoting its attenuation coefficient by $\eta_{\varepsilon}$, we obtain \begin{equation} 
	\begin{aligned} \mathscr F_K\!\left(\varepsilon(O\rho O^T)\right) = |\eta_{\varepsilon}| \mathscr F_K(O\rho O^T)= |\eta_{\varepsilon}|\mathscr F_K(\rho)= \mathscr F_K\!\left(\varepsilon(\rho)\right), 
	\end{aligned} \qquad K\in\{l_1,R\}.
\end{equation} 
This proves the first assertion. Finally, using the invariance of the input quantities and the equality of the corresponding output quantities, we find that \begin{equation}
	\begin{aligned} \Delta I_K^{\varepsilon}(O\rho O^T) &=\mathscr F_K(O\rho O^T) - \mathscr F_K\!\left(\varepsilon(O\rho O^T)\right)\\ &=\mathscr F_K(\rho) - \mathscr F_K\!\left(\varepsilon(\rho)\right)\\ &= \Delta I_K^{\varepsilon}(\rho), 
\end{aligned} \end{equation}
for $K\in\{l_1,R\}$. This completes the proof. 

\subsection{Proof of Proposition 4.3}\label{app:proof-operational-deimaginary}
Every $\Omega\in\mathcal M_{\mathrm{op}}^{\mathrm{sep}}$ is freely interconvertible with $\Omega_{++}$. Hence, by monotonicity, $\mathscr F_R(\Omega)=\mathscr F_r(\Omega)=1$. For a real Pauli operator $U\in\{X,Z\}$ and
$s\in[-1,1]$, define
\begin{align}
	\Phi_{s,U}(A)=\frac{1+s}{2}A+\frac{1-s}{2}UAU^{\dagger }.
\end{align}
The three channels can be written as
\begin{equation}
	\begin{aligned}
		\varepsilon_{\mathrm{PD}}=&\Phi_{\alpha_1,Z}\otimes\Phi_{\alpha_2,Z},\\
		\varepsilon_{\mathrm{PF}}=&\Phi_{t_1,Z}\otimes\Phi_{t_2,Z},\\
		\varepsilon_{\mathrm{BF}}=&\Phi_{t_1,X}\otimes\Phi_{t_2,X}.
	\end{aligned}
\end{equation}
For $s\neq0$, the inverse linear map is \begin{align}\Phi_{s,U}^{-1}(A)=\frac{1+s^{-1}}{2}A+\frac{1-s^{-1}}{2}UAU^{\dagger }.\end{align}
Since unitary conjugation preserves the trace norm, $\left\|\Phi_{s,U}^{-1}(A)\right\|_1\leq\frac{1}{|s|}\|A\|_1$, which implies $\left\|\Phi_{s,U}(A)\right\|_1\geq|s|\,\|A\|_1$.

Applying this inequality successively to the two local channels gives $\left\|\varepsilon(A)\right\|_1\geq c\,\|A\|_1$, where $c=c_{\mathrm{PD}}$ for the PD channel and
$c=c_{\mathrm F}$ for the PF and BF channels. The case $c=0$ follows directly by continuity.

Because these channels are real, $\varepsilon(\Omega)-\varepsilon(\Omega)^{\mathrm T}=\varepsilon\left(\Omega-\Omega^{\mathrm T}\right)$. Therefore, $\mathscr F_R\left(\varepsilon(\Omega)\right)\geq c\,\mathscr F_R(\Omega)=c$ for every
$\Omega\in\mathcal M_{\mathrm{op}}^{\mathrm{sep}}$. Thus, $\mathcal D_R^{\mathrm{op}}(\varepsilon)\leq1-c$.

By Lemma~\ref{lem:Fr-FR-relation}, $\mathscr F_r\left(\varepsilon(\Omega)\right)\geq1-\mathcal H\!\left(\mathscr F_R\left(\varepsilon(\Omega)\right)\right)$. Since $\mathcal H$ is decreasing on $[0,1]$ and $\mathscr F_R\left(\varepsilon(\Omega)\right)\geq c$, we obtain $\mathscr F_r\left(\varepsilon(\Omega)\right)\geq1-\mathcal H(c)$. Since every $\Omega\in\mathcal M_{\mathrm{op}}^{\mathrm{sep}}$ satisfies $\mathscr F_r(\Omega)=1$, it follows that $\mathscr F_r(\Omega)-\mathscr F_r\left(\varepsilon(\Omega)\right)\leq\mathcal H(c)$. Therefore, $\mathcal D_r^{\mathrm{op}}(\varepsilon)\leq\mathcal H(c)$.

To show that both bounds are attainable, set $P=X$ for the PD and PF channels, $P=Z$
for the BF channel, and define $\Omega_P^\star=\frac14\left(I\otimes I+Y\otimes P\right)$.

Let $\omega_\pm=\frac12(I\pm Y)$, $\pi_\pm^P=\frac12(I\pm P)$. Then $\Omega_P^\star=\frac12\left(\omega_+\otimes\pi_+^P+\omega_-\otimes\pi_-^P\right)$, which is a convex combination of two product states. Hence, $\Omega_P^\star$ is separable.

We next verify that it is operationally maximal. Let
$\pi_\pm^P=|p_\pm\rangle\langle p_\pm|$ and define the real Kraus
operators
\begin{align}
	K_+=I\otimes|0\rangle\langle p_+|,
	\qquad
	K_-=Z\otimes|0\rangle\langle p_-|.
\end{align}
They satisfy $K_+^\dagger K_++K_-^\dagger K_-=I_4$. Moreover, since $Z\omega_-Z=\omega_+$, the corresponding real channel gives
\begin{align}
	\sum_{\mu=\pm}
	K_\mu\Omega_P^\star K_\mu^\dagger=\omega_+\otimes|0\rangle\langle0|.
\end{align}
The real and imaginary parts of
$|+\rangle\otimes|0\rangle$ and
$|+\rangle\otimes|+\rangle$ form orthogonal pairs with the same norms. Therefore, there exists a real orthogonal matrix $O$ such that $O\left(|+\rangle\otimes|0\rangle\right)=|+\rangle\otimes|+\rangle$. Consequently, $\Omega_P^\star\longrightarrow\Omega_{++}$ by a real operation. Since $\Omega_{++}$ is operationally maximal, we conclude that $\Omega_P^\star\in \mathcal M_{\mathrm{op}}^{\mathrm{sep}}$.

For the corresponding channels, let
\[
s_i
=
\begin{cases}
	\alpha_i=\sqrt{1-\gamma_i},
	&
	\text{for the PD channel},
	\\
	t_i=2p_i-1,
	&
	\text{for the PF and BF channels},
\end{cases}
\]
and define $q=s_1s_2$, $c=|q|$. Since $Y\longmapsto s_1Y$, $P\longmapsto s_2P$, we obtain $\varepsilon(\Omega_P^\star)=\frac14\left(I_4+qY\otimes P\right)$. Using $Y^{\mathrm T}=-Y$, $P^{\mathrm T}=P$, $\|Y\otimes P\|_1=4$, we find $\mathscr F_R\left(\varepsilon(\Omega_P^\star)\right)=c$. Since $\mathscr F_R(\Omega_P^\star)=1$, the corresponding robustness loss is $1-c$.

Furthermore, $\operatorname{Re}\left[\varepsilon(\Omega_P^\star)\right]=\frac{I_4}{4}$, and the eigenvalues of $\varepsilon(\Omega_P^\star)$ are $\frac{1+c}{4}$, $\frac{1+c}{4}$, $\frac{1-c}{4}$, $\frac{1-c}{4}$. Therefore, $S\!\left(\operatorname{Re}\varepsilon(\Omega_P^\star)\right)=2$ and $S\!\left(\varepsilon(\Omega_P^\star)\right)=1+\mathcal H(c)$. It follows that $\mathscr F_r\left(\varepsilon(\Omega_P^\star)\right)=1-\mathcal H(c)$. Since $\mathscr F_r(\Omega_P^\star)=1$, the corresponding relative-entropy loss is $\mathcal H(c)$.

Thus, $\Omega_P^\star$ attains the previously established upper bounds. Hence,
\begin{align}
	\mathcal D_R^{\mathrm{op}}(\varepsilon)
	=
	1-c,
	\qquad
	\mathcal D_r^{\mathrm{op}}(\varepsilon)
	=
	\mathcal H(c).
\end{align}

\section{Kraus Operators of Two-Qubit Channels}\label{app:kraus-operators}
\renewcommand{\theequation}{B\arabic{equation}}
\renewcommand*{\theHequation}{B.\arabic{equation}}
\setcounter{equation}{0}
For completeness, we collect in this appendix the explicit Kraus representations of the two-qubit quantum channels used in the main text. These expressions are obtained from the corresponding single-qubit channels and are provided here to avoid interrupting the main discussion with lengthy matrix forms. In Section~\ref{sec:two-qubit-channel-analysis}, these Kraus representations are used to evaluate the corresponding channel actions: for the $PD$, $PF$, and $BF$ channels, they support the exact operational de-imaginary-power analysis, whereas for the remaining channels they are used to derive the reference-state losses.
\begin{equation}\label{eq:PD}
	\begin{aligned}
		K_{00}^{PD}=&\begin{bmatrix}\begin{smallmatrix}
				1&0  &  0&0 \\
				0&\sqrt{1-\gamma _{2} }   &  0&0 \\
				0&0  &\sqrt{1-\gamma _{1} }  &0 \\
				0&0 &0  &\sqrt{(1-\gamma _{1})(1-\gamma _{2})} 
		\end{smallmatrix}\end{bmatrix},\qquad\qquad
		K_{11}^{PD}=\begin{bmatrix}\begin{smallmatrix}
				0&0  &  0& 0\\
				0&0   &  0&0 \\
				0&0  &0  &0 \\
				0&0 &0  &\sqrt{\gamma _{1}\gamma _{2} } 
		\end{smallmatrix}\end{bmatrix},\\
		K_{01}^{PD}=&\begin{bmatrix}\begin{smallmatrix}
				0&0  &  0&0 \\
				0&\sqrt{\gamma _{2} }   &  0&0 \\
				0&0  &0  &0 \\
				0&0 &0  &\sqrt{\gamma _{2}(1-\gamma _{1})} 
		\end{smallmatrix}\end{bmatrix},\qquad\qquad\qquad\qquad\quad
		K_{10}^{PD}=\begin{bmatrix}\begin{smallmatrix}
				0&0  &0   &0 \\
				0&0   &  0&0 \\
				0&0  &\sqrt{\gamma _{1} }  &0 \\
				0&0 &0  &\sqrt{\gamma _{1} (1-\gamma _{2} )} 
		\end{smallmatrix}\end{bmatrix};		
	\end{aligned}
\end{equation}
\begin{equation}\label{eq:PF}
	\begin{aligned}
		K_{00}^{PF}=&\begin{bmatrix}\begin{smallmatrix}
				\sqrt{p_{1} p_{2} }&0  &  0&0 \\
				0&\sqrt{p_{1} p_{2} }   &  0&0 \\
				0&0  &\sqrt{p_{1} p_{2} }  &0 \\
				0&0 &0  &\sqrt{p_{1} p_{2} }
		\end{smallmatrix}\end{bmatrix},\qquad
		K_{11}^{PF}=\begin{bmatrix}\begin{smallmatrix}
				\sqrt{p_{1}^{\prime } p_{2}^{\prime }}&0  &  0&0  \\
				0&-\sqrt{p_{1}^{\prime } p_{2}^{\prime }}   &  0 &0 \\
				0&0  &-\sqrt{p_{1}^{\prime } p_{2}^{\prime }}  &0 \\
				0 &0 &0  &\sqrt{p_{1}^{\prime } p_{2}^{\prime }} 
		\end{smallmatrix}\end{bmatrix},\\
		K_{01}^{PF}=&\begin{bmatrix}\begin{smallmatrix}
				\sqrt{p_{1} p_{2}^{\prime } }&0  &  0&0 \\
				0&-\sqrt{p_{1} p_{2}^{\prime } }   &  0&0 \\
				0&0  &\sqrt{p_{1} p_{2}^{\prime } }  &0\\
				0&0 &0 &-\sqrt{p_{1} p_{2}^{\prime } }
		\end{smallmatrix}\end{bmatrix},
		K_{10}^{PF}=\begin{bmatrix}\begin{smallmatrix}
				\sqrt{p_{1}^{\prime } p_{2} }&0  &  0 &0 \\
				0&\sqrt{p_{1}^{\prime } p_{2} }   &  0&0 \\
				0&0  &-\sqrt{p_{1}^{\prime } p_{2} }  &0 \\
				0&0 &0  &-\sqrt{p_{1}^{\prime } p_{2} } 
		\end{smallmatrix}\end{bmatrix};		
	\end{aligned}
\end{equation}
\begin{equation}\label{eq:BF}
	\begin{aligned}
		K_{00}^{BF}=&\begin{bmatrix}\begin{smallmatrix}
				\sqrt{p_{1} p_{2} }&0  &  0&0 \\
				0&\sqrt{p_{1} p_{2} }   &  0&0 \\
				0&0  &\sqrt{p_{1} p_{2} }  &0 \\
				0&0 &0  &\sqrt{p_{1} p_{2} }
		\end{smallmatrix}\end{bmatrix},\qquad
		K_{11}^{BF}=\begin{bmatrix}\begin{smallmatrix}
				0&0  &  0&\sqrt{p_{1}^{\prime } p_{2}^{\prime }}  \\
				0&0   &  \sqrt{p_{1}^{\prime } p_{2}^{\prime }} &0 \\
				0&\sqrt{p_{1}^{\prime } p_{2}^{\prime }}  &0  &0 \\
				\sqrt{p_{1}^{\prime } p_{2}^{\prime }} &0 &0  &0 
		\end{smallmatrix}\end{bmatrix},\\
		K_{01}^{BF}=&\begin{bmatrix}\begin{smallmatrix}
				0&\sqrt{p_{1} p_{2}^{\prime } }  &  0&0 \\
				\sqrt{p_{1} p_{2}^{\prime } }&0   &  0&0 \\
				0&0  &0  &\sqrt{p_{1} p_{2}^{\prime } }\\
				0&0 &\sqrt{p_{1} p_{2}^{\prime } } &0
		\end{smallmatrix}\end{bmatrix},\quad\;
		K_{10}^{BF}=\begin{bmatrix}\begin{smallmatrix}
				0&0  &  \sqrt{p_{1}^{\prime } p_{2} } &0 \\
				0&0   &  0&\sqrt{p_{1}^{\prime } p_{2} } \\
				\sqrt{p_{1}^{\prime } p_{2} }&0  &0  &0 \\
				0&\sqrt{p_{1}^{\prime } p_{2} } &0  &0 
		\end{smallmatrix}\end{bmatrix};	
	\end{aligned}
\end{equation}
\begin{equation}\label{eq:BPF}
	\begin{aligned}
		K_{00}^{BPF}=&\begin{bmatrix}\begin{smallmatrix}
				\sqrt{p_{1} p_{2} }&0  &  0&0 \\
				0&\sqrt{p_{1} p_{2} }   &  0&0 \\
				0&0  &\sqrt{p_{1} p_{2} }  &0 \\
				0&0 &0  &\sqrt{p_{1} p_{2} }
		\end{smallmatrix}\end{bmatrix},\qquad\;
		K_{11}^{BPF}=\begin{bmatrix}\begin{smallmatrix}
				0&0  &  0&-\sqrt{p_{1}^{\prime }p_{2}^{\prime }}  \\
				0&0   &  \sqrt{p_{1}^{\prime }p_{2}^{\prime }} &0 \\
				0&\sqrt{p_{1}^{\prime }p_{2}^{\prime }}  &0  &0 \\
				-\sqrt{p_{1}^{\prime }p_{2}^{\prime }} &0 &0  &0 
		\end{smallmatrix}\end{bmatrix},\\
		K_{01}^{BPF}=&i\begin{bmatrix}\begin{smallmatrix}
				0&-\sqrt{p_{1} p_{2}^{\prime } }  &  0&0 \\
				\sqrt{p_{1} p_{2}^{\prime } }&0   &  0&0 \\
				0&0  &0  &-\sqrt{p_{1} p_{2}^{\prime } }\\
				0&0 &\sqrt{p_{1} p_{2}^{\prime } } &0
		\end{smallmatrix}\end{bmatrix},
		K_{10}^{BPF}=i\begin{bmatrix}\begin{smallmatrix}
				0&0  &  -\sqrt{p_{1}^{\prime } p_{2} } &0 \\
				0&0   &  0&-\sqrt{p_{1}^{\prime } p_{2} } \\
				\sqrt{p_{1}^{\prime } p_{2} }&0  &0  &0 \\
				0&\sqrt{p_{1}^{\prime } p_{2} } &0  &0 
		\end{smallmatrix}\end{bmatrix},		
	\end{aligned}
\end{equation}
where $p_{1}^{\prime }=1-p_{1}$, $p_{2}^{\prime }=1-p_{2}$.

\begin{equation}\label{eq:AD}
	\begin{aligned}
		K_{00}^{AD}=&\begin{bmatrix}\begin{smallmatrix}
				1&0  &  0&0 \\
				0&\sqrt{1-\gamma _{2} }   &  0&0 \\
				0&0  &\sqrt{1-\gamma _{1} }  &0 \\
				0&0 &0  &\sqrt{(1-\gamma _{1})(1-\gamma _{2})} 
		\end{smallmatrix}\end{bmatrix},\qquad\qquad
		K_{11}^{AD}=\begin{bmatrix}\begin{smallmatrix}
				0&0  &  0&\sqrt{\gamma _{1}\gamma _{2} } \\
				0&0   &  0&0 \\
				0&0  &0  &0 \\
				0&0 &0  &0 
		\end{smallmatrix}\end{bmatrix},\\
		K_{01}^{AD}=&\begin{bmatrix}\begin{smallmatrix}
				0&\sqrt{\gamma _{2} }  &  0&0 \\
				0&0   &  0&0 \\
				0&0  &0  &\sqrt{\gamma _{2}(1-\gamma _{1})} \\
				0&0 &0  &0 
		\end{smallmatrix}\end{bmatrix},\qquad\qquad\qquad\quad\qquad
		K_{10}^{AD}=\begin{bmatrix}\begin{smallmatrix}
				0&0  &  \sqrt{\gamma _{1} } &0 \\
				0&0   &  0&\sqrt{\gamma _{1} (1-\gamma _{2} )} \\
				0&0  &0  &0 \\
				0&0 &0  &0 
		\end{smallmatrix}\end{bmatrix};			
	\end{aligned}
\end{equation}
\begin{equation}\label{eq:PAD}
	\begin{aligned}
		K_{00}^{PAD}=&\begin{bmatrix}\begin{smallmatrix}
				1&0  &  0&0 \\
				0&\sqrt{1-\gamma _{2} }   &  0&0 \\
				0&0  &\sqrt{1-\gamma _{1} }  &0 \\
				0&0 &0  &\sqrt{(1-\gamma _{1})(1-\gamma _{2})} 
		\end{smallmatrix}\end{bmatrix},\qquad\qquad
		K_{11}^{PAD}=\begin{bmatrix}\begin{smallmatrix}
				0&0  &  0&0 \\
				0&0   &  0&0 \\
				0&0  &0  &\sqrt{\gamma _{1}\gamma _{2} } \\
				0&0 &0  &0 
		\end{smallmatrix}\end{bmatrix},\\
		K_{01}^{PAD}=&\begin{bmatrix}\begin{smallmatrix}
				0&\sqrt{\gamma _{2} }  &  0&0 \\
				0&0   &  0&0 \\
				0&0  &0  &\sqrt{\gamma _{2}(1-\gamma _{1})} \\
				0&0 &0  &0 
		\end{smallmatrix}\end{bmatrix},\qquad\qquad\qquad\quad\qquad
		K_{10}^{PAD}=\begin{bmatrix}\begin{smallmatrix}
				0&0  &  0 &0 \\
				0&0   &  0&0 \\
				0&0  &\sqrt{\gamma _{1} }  &0 \\
				0&0 &0  &\sqrt{\gamma _{1} (1-\gamma _{2} )}
		\end{smallmatrix}\end{bmatrix}.		
	\end{aligned}
\end{equation}

\section{Explicit Expressions of Evolved Density Matrices under Various Quantum Channels}\label{app:evolved-density-matrices}
\renewcommand{\theequation}{C\arabic{equation}}
\renewcommand*{\theHequation}{C.\arabic{equation}}
\setcounter{equation}{0}
In this appendix, we present the explicit density matrices obtained by applying the two-qubit channels considered in Section~\ref{sec:two-qubit-channel-analysis} to the operationally maximal reference state $\Omega_{++}$. These expressions are used to derive the corresponding reference-state losses. For the PD, PF, and BF channels, they also provide the reference-state benchmarks that are compared with the exact operational results obtained in the main text.
\begingroup
\small
\begin{equation}\label{eq:pd-state-evol}
		\begin{aligned}
			\varepsilon _{PD}(\ket{+}\bra{+}\otimes\ket{+}\bra{+})=&\frac{1}{4}\begin{bmatrix}\begin{smallmatrix}
					1 & -i\sqrt{1 - \gamma_{2}} & -i\sqrt{1 - \gamma_{1}} & -\sqrt{(1 - \gamma_{1})(1 - \gamma_{2})}\\
					i\sqrt{1 - \gamma_{2}} & 1 & \sqrt{(1 - \gamma_{2})(1 - \gamma_{1})} & -i\sqrt{1 - \gamma_{1}}\\
					i\sqrt{1 - \gamma_{1}} & \sqrt{(1 - \gamma_{2})(1 - \gamma_{1})} & 1 & -i\sqrt{1 - \gamma_{2}}\\
					-\sqrt{(1 - \gamma_{1})(1 - \gamma_{2})} & i\sqrt{1 - \gamma_{1}} & i\sqrt{1 - \gamma_{2}} & 1
			\end{smallmatrix}\end{bmatrix};
	\end{aligned}
\end{equation}\endgroup
\begingroup
\small
\begin{equation}\label{eq:pf-state-evol}
		\begin{aligned}
			\varepsilon _{PF}(\ket{+}\bra{+}\otimes\ket{+}\bra{+})= &\frac{1}{4} \begin{bmatrix}\begin{smallmatrix}
					1&-i(2p_2 - 1)&-i(2p_1 - 1)&-(2p_1 - 1)(2p_2 - 1)\\
					i(2p_2 - 1)&1&(2p_1 - 1)(2p_2 - 1)&-i(2p_1 - 1)\\
					i(2p_1 - 1)&(2p_1 - 1)(2p_2 - 1)&1&-i(2p_2 - 1)\\
					-(2p_1 - 1)(2p_2 - 1)&i(2p_1 - 1)&i(2p_2 - 1)&1
			\end{smallmatrix}\end{bmatrix};
	\end{aligned}
\end{equation}\endgroup
\begingroup
\small
\begin{equation}\label{eq:bf-state-evol}
		\begin{aligned}
			\varepsilon _{BF}(\ket{+}\bra{+}\otimes\ket{+}\bra{+})= &\frac{1}{4} \begin{bmatrix}\begin{smallmatrix}
					1&-i(2p_2 - 1)&-i(2p_1 - 1)&-(2p_1 - 1)(2p_2 - 1)\\
					i(2p_2 - 1)&1&(2p_1 - 1)(2p_2 - 1)&-i(2p_1 - 1)\\
					i(2p_1 - 1)&(2p_1 - 1)(2p_2 - 1)&1&-i(2p_2 - 1)\\
					-(2p_1 - 1)(2p_2 - 1)&i(2p_1 - 1)&i(2p_2 - 1)&1
			\end{smallmatrix}\end{bmatrix}.
	\end{aligned}
\end{equation}\endgroup
\begingroup
\small
\begin{equation}\label{eq:ad-state-evol}
		\begin{aligned}
			\varepsilon _{AD}(\ket{+}\bra{+}\otimes\ket{+}\bra{+})
			= &\frac{1}{4} \begin{bmatrix}\begin{smallmatrix}
					1+\gamma _{1}+\gamma _{2}+\gamma _{1}\gamma _{2} &-i(1+\gamma _{1} )\sqrt{1-\gamma _{2} }  &-i(1+\gamma _{2})  \sqrt{1-\gamma _{1} }&-\sqrt{(1-\gamma _{1})(1-\gamma _{2})} \\
					i(1+\gamma _{1} )\sqrt{1-\gamma _{2} }&(\gamma _{1}+1) (1-\gamma _{2}) &\sqrt{(1-\gamma _{1})(1-\gamma _{2})}  &-i(1-\gamma _{2})\sqrt{1-\gamma _{1} } \\
					i(1+\gamma _{2})  \sqrt{1-\gamma _{1} }&\sqrt{(1-\gamma _{1})(1-\gamma _{2})}  &(1-\gamma _{1})(1+\gamma _{2})  &-i(1-\gamma _{1})\sqrt{1-\gamma _{2} } \\
					-\sqrt{(1-\gamma _{1})(1-\gamma _{2})} &i(1-\gamma _{2})\sqrt{1-\gamma _{1} }  &i(1-\gamma _{1})\sqrt{1-\gamma _{2} }  &(1-\gamma _{1})(1-\gamma _{2})
			\end{smallmatrix}\end{bmatrix} ,
		\end{aligned}
\end{equation}\endgroup
\begingroup
\small
\begin{equation}\label{eq:pad-state-evol}
		\begin{aligned}
			\varepsilon _{PAD}(\ket{+}\bra{+}\otimes\ket{+}\bra{+})
			= &\frac{1}{4} \begin{bmatrix}\begin{smallmatrix}
					1 + \gamma_2 & -i\sqrt{1 - \gamma_2} & -i(1 + \gamma_2)\sqrt{1 - \gamma_1} & -\sqrt{(1 - \gamma_1)(1 - \gamma_2)}\\
					i\sqrt{1 - \gamma_2} & 1 - \gamma_2 & \sqrt{(1 - \gamma_1)(1 - \gamma_2)} & -i(1 - \gamma_2)\sqrt{1 - \gamma_1}\\
					i(1 + \gamma_2)\sqrt{1-\gamma_1} & \sqrt{(1 - \gamma_1)(1 - \gamma_2)} & 1 + \gamma_2 & -i\sqrt{1 - \gamma_2}\\
					-\sqrt{(1 - \gamma_1)(1 - \gamma_2)} & i(1 - \gamma_2)\sqrt{1-\gamma_1} & i\sqrt{1 - \gamma_2} & 1 - \gamma_2
			\end{smallmatrix}\end{bmatrix}.
		\end{aligned}
\end{equation}\endgroup

\end{document}